\documentclass[sigconf,authorversion,nonacm]{acmart}
\usepackage[T1]{fontenc}
\usepackage[utf8]{inputenc}
\usepackage{graphicx,tikz,xcolor,hyperref,booktabs,upquote,csquotes,textcomp,subfig,multirow,paralist}
\usepackage[ruled,vlined,longend,linesnumbered]{algorithm2e}
\usepackage{algorithmicx}
\usepackage{pifont}
\usepackage{xcolor,colortbl}
\usepackage{xspace}
\usepackage{wrapfig}
\usepackage[inline]{enumitem}

\usetikzlibrary{arrows,decorations.pathmorphing,backgrounds,positioning,fadings,shadings,through,fit,petri,calc}

\makeatletter
\newcommand{\removelatexerror}{\let\@latex@error\@gobble}
\makeatother

\newtheorem{theorem}{Theorem}

\newtheorem{corollary}{Corollary}

\newcolumntype{P}[1]{>{\raggedright\arraybackslash}p{#1}}
\newcolumntype{L}[1]{>{\raggedright\let\newline\\\arraybackslash\hspace{0pt}}m{#1}}
\newcolumntype{C}[1]{>{\centering\let\newline\\\arraybackslash\hspace{0pt}}m{#1}}
\newcolumntype{R}[1]{>{\raggedleft\let\newline\\\arraybackslash\hspace{0pt}}m{#1}}

\newcommand{\NM}{(M+N')}

\newcommand{\N}{\mathbb{N}}
\newcommand{\nodes}{\mathcal{N}}

\newcommand{\lines}{\mathcal{L}}
\newcommand{\meters}{\mathcal{M}}
\newcommand{\owners}{\mathcal{O}}
\newcommand{\nb}{{N}_{\textsc{nb}}}
\newcommand{\measurements}{\mathcal{X}}

\newcommand{\vecx}{\vec{x}}

\newcommand{\vece}{\vec{\epsilon}}
\newcommand{\vect}{\vec{\theta}}
\newcommand{\matb}{G}

\newcommand{\largeconstant}{\kappa}
\newcommand{\initdeposit}{\gamma}
\newcommand{\timeoutblocks}{\tau^*}
\newcommand{\slotduration}{\eta}
\newcommand{\timeout}{\tau}
\newcommand{\reward}{\rho}
\newcommand{\misspen}{\phi}
\newcommand{\anompen}{\phi'}
\newcommand{\missprob}{p}
\newcommand{\anomprob}{p'}
\newcommand{\othersprob}{q}
\newcommand{\anomthres}{\theta}

\newcommand{\rtot}{R_o^{\textsc{tot}}}
\newcommand{\roth}{R_o^{\textsc{oth}}}
\newcommand{\phitot}{\Phi_o^{\textsc{tot}}}
\newcommand{\phioth}{\Phi_o^{\textsc{oth}}}
\newcommand{\phip}{\Phi_o'}
\newcommand{\rnet}{R_o^{\textsc{net}}}

\renewcommand{\i}{{\bf 1}}
\newcommand{\iadd}{\i^{\textsc{add}}}
\newcommand{\imiss}{\i^{\textsc{miss}}}
\newcommand{\iiadd}{{\bf I}}
\newcommand{\iimiss}{{\bf I}}
\newcommand{\distortmeters}{\mathcal{D}}
\newcommand{\maker}{{\bf M}}
\newcommand{\makerentry}{{\bf m}}

\begin{document}

\author{Dani\"el Reijsbergen, Aung Maw, Tien Tuan Anh Dinh, Wen-Tai Li, and Chau Yuen}
\affiliation{Singapore University of Technology and Design}

\title{Securing Smart Grids Through an Incentive Mechanism for Blockchain-Based Data Sharing}

\begin{abstract}
Smart grids leverage the data collected from smart meters to make important operational decisions. However, they are vulnerable to False Data Injection (FDI) attacks in which an attacker manipulates meter data to disrupt the grid operations. Existing works on FDI are based on a simple threat model in which a single grid operator has access to all the data, and only some meters can be compromised. 

Our goal is to secure smart grids against FDI under a realistic threat model. To this end, we present a threat model in which there are multiple operators, each with a partial view of the grid, and each can be fully compromised. An effective defense against FDI in this setting is to share data between the operators. However, the main challenge here is to incentivize data sharing. We address this by proposing an incentive mechanism that rewards operators for uploading data, but penalizes them if the data is missing or anomalous. We derive formal conditions under which our incentive mechanism is provably secure against operators who withhold or distort measurement data for profit. We then implement the data sharing solution on a private blockchain, introducing several optimizations that overcome the inherent performance limitations of the blockchain. Finally, we conduct an experimental evaluation that demonstrates that our implementation has practical performance.  
\end{abstract}

\maketitle

\section{Introduction}
\label{sec:intro}
Meters in modern power systems are increasingly augmented with automated data collection and communication
capabilities, leading to the emergence of smart grids.  The high-frequency measurements gathered by
these \textit{smart meters} allow for unprecedented levels of feedback and control.  However, as utility companies
become more reliant on smart meter data for grid management, the meters and their data become more attractive
targets for attackers. As a critical infrastructure, power systems are at risk from, e.g., hackers seeking to
extort utility companies \cite{ukattack}, or state-sponsored adversaries seeking to disrupt a power grid for
political reasons \cite{ukraine2016analysis}.  

One important attack against smart power grids is False Data Injection
(FDI)~\cite{bobba2010detecting,lakshminarayana2020data,lakshminarayana2017optimal,liu2011false,liu2015collaborative,wang2019online}.
In an FDI attack, the measurements collected by the meter during normal operation are replaced by arbitrary
measurements, which violates data integrity.  Existing approaches against FDI include redundancy through the
strategic placement of additional meters
\cite{chen2006placement,rakpenthai2006optimal,denegri2002security,yang2017optimal}, improving the security of
the existing meters \cite{jamei2016micro,mazloomzadeh2015empirical}, and robust detection algorithms
\cite{wang2019online}.  However, these works consider a simple threat model in which there is a single operator
who has access to all the measurement data. Furthermore, they assume that the attackers are limited in the
number of meters that they can compromise, or in the magnitude with which they can distort measurements.  

We seek a solution for protecting smart grids against FDI under a realistic threat model. We
note that the model used in existing works precludes attacks that overwrite the data from \textit{all} of a single operator's meters.  However, such
an attack cannot be ruled out, e.g., if the
storage servers at the operator are compromised by a malicious insider or after a successful hack.  In this paper, we consider a more realistic model, in which there are  multiple
operators, and each of them only has a partial view of the grid. Each operator can be fully compromised, i.e., all of
its meters and its data management servers can be under the attacker's control.  

Under this new threat model, we propose a solution for defending against FDI in which operators form a
\textit{consortium} to share their data. In particular, the operators first agree on a high-level {system topology} and
 indicate the location of their meters within this topology. They then upload meter measurements to a shared, tamper-evident data structure --
e.g., a blockchain. The operators have access to the full measurement dataset, and can therefore run
  their own FDI detection algorithms. We focus on meters on the higher levels of the grid -- i.e., the distribution and transmission stages -- as they have the following advantages over household meters: 1) user privacy is less of a concern because each high-level measurement aggregates many households, 2) high-level meters perform and share measurements at a higher frequency than household meters, and 3) the number of high-level meters in a regional or national network is typically in the order of thousands (e.g., see the IEEE test cases for Poland and the Western USA \cite{ieeedataset}).

The main motivation for joining our data sharing solution is to defend against FDI attacks.
However, we note that convincing the operators to join the consortium is not
sufficient. The grid operators are competitors in the same market as they bid for the same contracts to operate sections of the grid, and one operator may therefore profit if the others do
poorly. This implies that the operators may be motivated to share bad or insufficient data.   
As an example, consider an operator who shares only a small fraction of its meters' data. If its goal is to
determine system-wide measurement averages, it can first obtain data from other operators, then combine these 
with its hidden measurements. This way, the operator achieves its goal while reducing the degree to which its
competitors benefit from the system. A malicious operator can even share distorted measurements in order to
mislead other operators. 

We address this incentivization challenge by proposing an \textit{incentive mechanism} in which operators make an initial financial deposit, and earn \textit{credits} when they share meaningful data. Moreover, they lose credits if they fail to share data or share anomalous data. An operator who loses all of its credits after incurring penalties over a long period (e.g., weeks) can be expelled from the consortium, causing it to lose its deposit. To identify anomalies, we
exploit the interdependence between measurements: in particular, a
measurement is flagged as \textit{anomalous} if it is found, through state estimation \cite{monticelli1999state,abur2004power}, to be inconsistent with the measurements from other upstream or
downstream meters. 
Our incentive mechanism helps operators to identify and fix benign problems that cause their
meters to report anomalous data. It also helps to identify bad operators who persistently share bad data. 
Furthermore, we  derive formal conditions under which our mechanism is provably \textit{incentive compatible}: i.e., it is the optimal strategy for (coalitions of) operators to share as many accurate measurements as possible unless certain technical conditions are violated. 
These conditions depend on the structure of the grid and the strength of the adversarial coalition.

We implement our data sharing solution on top of a private \textit{blockchain}. Blockchains allow us to implement credit transfers between parties in a secure, transparent, and fully automated manner on a network of mutually untrusting peers \cite{hassan2019blockchain}. Blockchains remove the need for a trusted third party, which itself may suffer from attacks or temporary losses of service that impact the incentive mechanism, leading to costly disputes. We implement the mechanism in
a \textit{smart contract} that performs state estimation, which requires expensive matrix computations. Since smart contract execution times
affect the blockchain's throughput, the challenge here is to \textit{speed up} such execution in order to avoid
throughput degradation. As a final contribution, we therefore introduce several performance improvements: data storage optimization, moving expensive computations offline, utilizing a randomized algorithm to verify offline computations \cite{freivalds1979fast,tramer2018slalom}, and multithreading. Our experiments show that our approach achieves practical performance.

In summary, we make the following four contributions: 
\begin{enumerate}
\item We present a threat model for FDI attacks that is broader and more realistic than the threat models in existing works because it includes the possibility that \textit{all} of a single operator's measurements are affected.

\item We present a data sharing solution that defends against FDI attacks in this threat model if other operators' meters perform redundant measurements.

\item Our data sharing solution features an incentive mechanism that is provably incentive-compatible, i.e., it motivates operators to share meaningful data, under explicit technical conditions that depend on the grid topology and the strength of the adversarial coalition.

\item We implement our solution on a private blockchain and introduce four optimizations that speed up storage
operations and matrix computation in the contract. We conduct an experimental evaluation of our implementation
and show that it achieves practical performance.  
\end{enumerate}

The rest of the paper is organized as follows. In Section~\ref{sec:background}, we review the relevant background
on smart grids and blockchains. In Section~\ref{sec:model}, we describe the system model, threat model, and the requirements for our solution.
We discuss our incentive mechanism in Section~\ref{sec:incentives} and the implementation of our data sharing solution in Section~\ref{sec:implementation}. 
We analyze our solution in Section~\ref{sec:analysis} and find that it satisfies the stated requirements. 
We present the empirical evaluation in Section~\ref{sec:experiments}, discuss related work in Section~\ref{sec:related}, and conclude the paper in Section~\ref{sec:conclusion}.  

\section{Background}
\label{sec:background}

\subsection{Power Grid}
Power grids are systems that deliver electrical power from power plants to consumers such as households and
businesses. Figure~\ref{fig:meter} shows the design of a grid which consists of four stages:
\textit{generation}, \textit{transmission}, \textit{distribution}, and \textit{consumption.} After generation,
power is transmitted over a network of high-voltage power lines, after which it is distributed to consumers
over medium- and low-voltage networks.  The power lines and nodes (junctions and substations) in the
transmission network are operated by \textit{Transmission System Operators} (TSOs).  Similarly, the power
lines and nodes in the distribution network are operated by \textit{Distribution System Operators} (DSOs).
Billing and customer service is
provided to consumers by \emph{retailers}.
\begin{figure}
\centering
\includegraphics[width=0.4\textwidth]{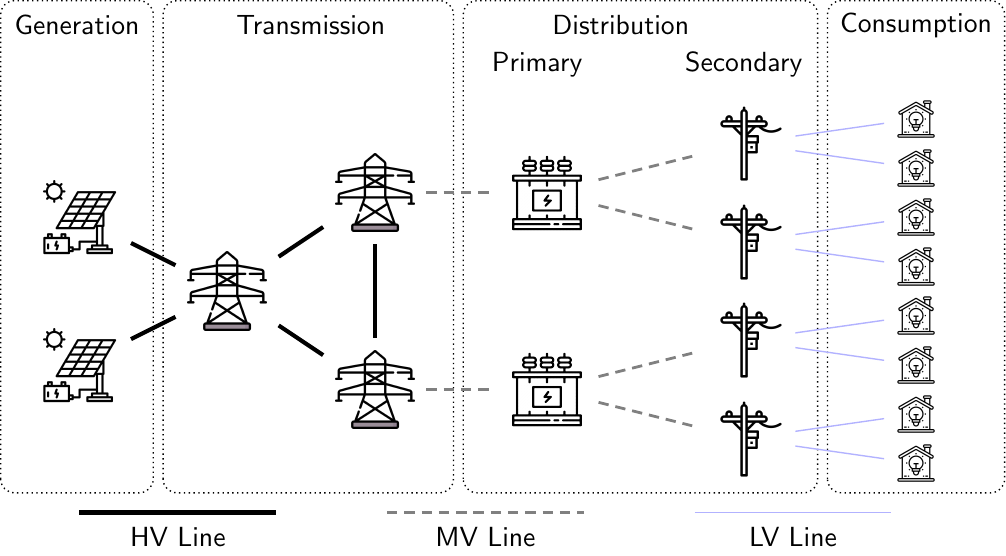}
\caption{High-level overview of the power grid.}
\label{fig:meter}
\end{figure}
The roles of different organizations in the grid vary from country to country.  For example, in Japan, there are 10 regions in which a single operator acts as a TSO and DSO \cite{tepco}. However, in Germany, the national grid is operated by four
TSOs, but the distribution networks are operated by hundreds of municipal utility companies
(\textit{Stadtwerke}) and some regional DSOs \cite{germangrid}. 

\begin{figure}
\centering
\includegraphics[width=0.4\textwidth]{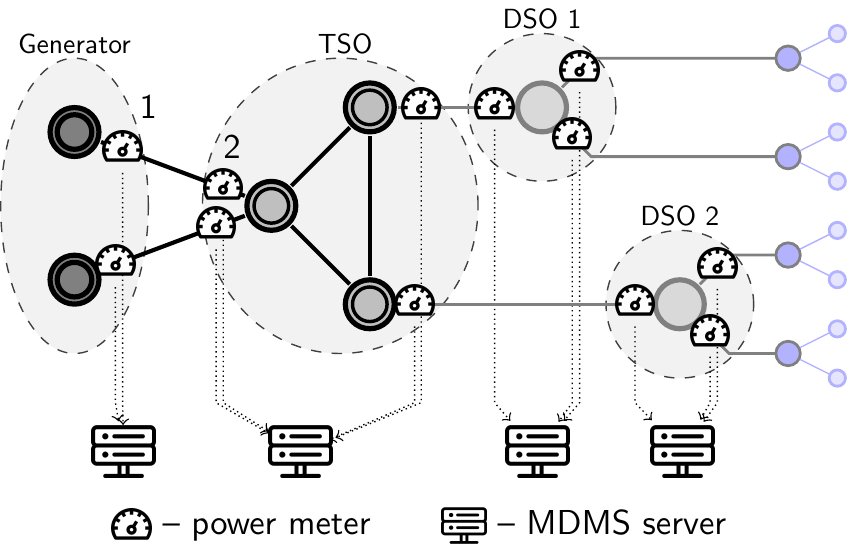}
\caption{Meters and data collection in the grid. Solid lines represent power lines, whereas dotted lines
indicate data transfers. There are four operators (one generator, one TSO, and two DSOs), and each has
its own MDMS to collect and process data from their meters. By sharing their data, the operators are able to
fully exploit the redundancy in the grid, e.g., between the two meters marked `1' and `2'.}
\label{fig:solution} \end{figure}

In each stage of the grid, the electricity going over the power lines is measured using \emph{meters}. In this
paper, we focus particularly on the meters at the \textit{higher levels} of the grid, that is, at the generation,
transmission, and distribution stages.  These meters are an essential component of the operator's Supervisory
Control And Data Acquisition (SCADA) system. The main goal of the SCADA system is to facilitate operations
management such as load balancing because excessively high loads can cause line failures and blackouts. SCADA
meters can measure a number of different quantities, for example, voltage, current, and frequency. However,
the main quantities of interest in our setting are the \textit{power flow} and \textit{injection} on lines and
nodes, respectively. In addition to meters on the power lines, grid operators can also deploy high-precision meters
called Phasor Measurement Units (PMUs) at the nodes. A PMU outputs high-frequency estimates of
voltage and current phasors ~\cite{jamei2016micro,rakpenthai2006optimal,wang2019online,li2015location}. 
The high-level meters all have advanced communication capabilities. Meter data is
collected, stored, and processed by the operator's Meter Data Management System (MDMS).

The data collected from the meters allows for the estimation of system states (i.e., voltage and current
phasors) at the nodes. In the following, we use the bus-branch model, in which the branches represent lines
and the buses represent nodes, in combination with the DC flow model
\cite{abur2004power,gomez2018electric,wood2013power,monticelli1999state}.  In this model, the relationship
between the measurements and the system states is linear. Furthermore, measurement inaccuracies -- e.g., due to different meters not performing measurements at exactly the same time -- are expressed as random variables with a known distribution (e.g., normal/Gaussian with known variance). 
To detect faults or attacks, we can then use \textit{linear regression}: we use the least-squares method to estimate the system states,
and then use the result to re-estimate the measurements. If the sum of the differences between the observed
and the estimated measurements, called the \textit{residuals}, is significantly higher than what we would expect from the probability distribution of the measurement errors, then this is evidence of a fault or an attack. 

Depending on the placement of the meters, the system can be \textit{observable} or
\textit{unobservable} \cite{abur2004power,bobba2010detecting}. The system is unobservable if the matrix used for
state estimation does not have full rank. Individual meters can be either \textit{critical} or
\textit{redundant}. A meter is critical if the following holds: if the meter is included then the
system is observable, but if the meter is removed then the system is unobservable. A meter that is not
critical is redundant. A residual that corresponds to a critical meter will always be zero, but those that correspond to a redundant meter will typically be non-zero. 
When a single redundant meter produces a bad measurement (i.e., different from the true quantity)
then this meter will have the largest residual \cite{abur2004power}. 

Figure~\ref{fig:solution} illustrates an example of different operators managing different sets of meters.
Meters 1 and 2 are redundant, as they are on the same line and are likely to produce similar measurements.
This redundancy in the grid can only be exploited if the operators share measurement data. Some countries, such as the USA, make it compulsory for
grid operators to share data with a regulator. There are even standard data formats for exchanging information
between grid operators.\footnote{\url{https://www.entsoe.eu/digital/common-information-model/\#common-information-model-cim-for-grid-models-exchange}} However, we do not assume that the regulator has an active role in monitoring grid measurements. The potential role of the regulator will be discussed in more detail in Section~\ref{sec:analysis}.

\subsection{Blockchains}
A blockchain is a distributed system consisting of mutually distrusting peers. The peers maintain some
global states which can be updated via \textit{transactions}. The blockchain groups transactions into \textit{blocks} -- each
block links to another block via a cryptographic hash pointer, forming a chain. A Merkle tree \cite{merkle1987digital} is built
on the global states, and the root of the tree is included in each block. The nodes run a consensus protocol
to ensure that each node has the same copy of the states.   

There are two main types of blockchains: public (or permissionless) and private (or
permissioned)~\cite{dinh2017blockbench}. The former type allows any node to join the network, and typically uses a variant of
 Proof-of-Work~\cite{bitcoin} for consensus. The latter is restricted to approved members, and
uses a classic Byzantine fault-tolerant consensus protocol such as PBFT~\cite{castro1999practical}. One popular private
blockchain is Hyperledger Fabric~\cite{hyperledger}. It supports a centralized \textit{membership service} that maintains a list of users
who can read and write to the blockchain. 

Many blockchain systems today support {\em smart contracts} which are pieces of code that run on the blockchain. The
smart contract can access the global states, and its execution is replicated at all nodes. Some blockchains
allow Turing-complete smart contracts. For example, Ethereum \cite{ethereum} has its own language and executes the code on the
Ethereum Virtual Machine (EVM). Hyperledger Fabric~\cite{hyperledger} runs its smart contracts inside a Docker container. 

\section{Model \& Requirements}
\label{sec:model}
\subsection{System Model}
The main entities in our system model are the grid operators, i.e., the TSOs and DSOs. We denote the set of operators by
$\owners = \{1,\ldots,O\}$. The grid is represented by a graph that consists of nodes and power lines. Let $\nodes =
\{1,\ldots,N\}$ be the set of nodes and $\lines  = \{1,\ldots,L\}$ the set of lines. 
As discussed in
Section~\ref{sec:background}, we consider three types of meters: power flow meters, node injection
meters, and PMUs. 
We denote the set of meters by $\meters  =
\{1,\ldots,M\}$. 
Some nodes $\nodes' \subset \nodes$
are assumed to always have zero power injection. The set $\meters$ is partitioned in the following way: for each operator $o \in \owners$, let
$\meters_o \subset \meters$ be the set of meters owned by $o$. Let $N'$ and $M_o$ the size of $\nodes'$ and $\meters_o$, respectively. Time is represented as a discrete set of \textit{slots} $t\in \{1,2,\ldots\}$. Each meter $m \in
\meters$ is expected to produce a measurement $x_{m\,t}$ in each time slot $t$, and send it to its operator's MDMS
server for processing.
We assume that time slots have fixed start and end points and that the operators' servers have loosely synchronized clocks. In the following, when we consider a single time slot, we omit the subscript $t$ and simply write $x_m$.

In practice, measurements can be delayed or missing altogether. The processing of measurements may be delayed for various reasons, including: 
\begin{inparaenum}[1)]
\item hardware or software problems on the meter affecting the transmission of the measurement to the operator's
MDMS server,
\item network conditions between the meter and MDMS server, or
\item hardware or software problems on the MDMS server.  
\end{inparaenum}
Similarly, reasons for measurements that permanently go missing include:
\begin{inparaenum}[1)]
\item hardware or software problems on the meter that prevent measurements from being recorded,
\item hardware or software problems on the meter that delete measurements stored on the flash memory before
they are transmitted,
\item messages being lost during transmission, or
\item hardware or software problems on the MDMS server that delete the measurements. 
\end{inparaenum}

It is not feasible to wait indefinitely for measurements since they may have gone missing. In addition, 
measurements from too far in the past are no longer relevant for operations management. We model this through a \textit{cut-off time} $\timeout$, which is the number of time slots after which measurements are no
longer relevant. We assume that each meter $m \in \meters$ is \textit{offline} in each time slot with a fixed probability $\missprob_m$.
In our setting, offline means that the meter is not able to communicate with the MDMS server from the start of the time slot until the end of the cut-off period. 
A meter can also produce \textit{bad data}, which we define as data that is, with high
probability (given the probability distribution of the measurement errors), different from the true quantity that is being measured. 

Given a grid specified in terms of $\nodes$, $\nodes'$, $\lines$, $\meters$, we can perform state
estimation as follows. First, we construct a topology matrix $G$ such that $\vec{x} = G
\vec{\theta} + \vec{\epsilon}$ where $\vec{x}$ is a column vector with measurements, $\vec{\theta}$ is the voltage phasor
angle in each of the nodes, and $\vec{\epsilon}$ are the errors. 
Next, we determine
the projection matrix $P=G(G^{T}G)^{-1}G^{T}$ for the least squares method. The residuals are then computed as
\mbox{$\hat{\epsilon} = \vec{x} - P\vec{x}$}.
If the sum of squared residuals, given by $\vec{\epsilon^{T} \epsilon} = \sum_{m \in \meters} \epsilon_m^2$, exceeds a given threshold $\theta'$, then an anomaly -- possibly caused by an FDI attack -- is detected. A more detailed description of how to construct the matrix $G$ using the grid topology can be found in Algorithm~\ref{sec:grid_matrix}.

\subsection{Threat Model}
\label{sec:threat_model}

In our model, we consider \textit{two} distinct types of attackers.

\textit{FDI Attacker: } this attacker's goal is to \textit{cause short-term disruption} to grid operations, potentially causing a blackout. FDI attackers can be very powerful, i.e., state-sponsored. We assume that this attacker is able to compromise some, but not all operators for a short period of time. By ``compromise'' we mean that the attacker has obtained access to the operator's computer systems or the private keys or passwords stored on these systems, e.g., through a hack, (spear) phishing, or an insider threat. 

\textit{Persistently Adversarial Operator: } this attacker's aim is to \textit{maximize its long-term profits}, possibly by deviating from the protocol. 
We consider the following types of persistent adversarial behavior by an operator: 1) refraining from registering valid meters, 2) uploading fake measurements from offline or non-existing meters, 3) taking its own meters offline by disconnecting them from the network, 4) preventing another operator's meters from sharing measurements, and 5) uploading distorted measurements by manipulating input signals or placing the meters at a different location in the grid than what is claimed.

An overview of threat models in the related literature can be found in Appendix~\ref{sec:threat_model_table}.

\subsection{Requirements}
\label{sec:requirements}
In the remainder of this work, we present a solution that satisfies the following requirements. 
\begin{enumerate}
\item \textbf{FDI Attack Security}. FDI attacks against compromised operators should still be detectable if redundant measurements are performed by uncompromised operators. 
\item \textbf{Incentive Compatibility.} Operators should be incentivized to join the system, and to refrain from the persistent adversarial behavior discussed in the previous section. 
\item \textbf{Scalability.} The solution should be able to support thousands of meters.
\end{enumerate}
In Sections~\ref{sec:incentives}~and~\ref{sec:implementation}, we present our data sharing mechanism and the implementation of our data sharing solution, respectively. We discuss under what conditions our solution satisfies the above requirements in Section~\ref{sec:analysis}.

\section{Incentive Mechanism}
\label{sec:incentives}
As discussed in Section~\ref{sec:model}, an FDI attacker who successfully compromises an operator can tamper with all of that operator's readings. Such attacks can be mitigated by sharing data between operators, but a na\"ive data sharing solution is vulnerable to adversarial behavior by operators. Grid operators compete with each other for the same customers, and as such may be willing to take
actions that harm other operators if it comes at no cost to themselves. As an example, see the recent EU lawsuit in
which TenneT \cite{eutennet}, which is both a TSO and a generator, was accused of restricting grid capacity to
 competing generators. Under our threat model, an adversarial strategy could arise in which an operator
withholds its own data while still reading the data from others. Alternatively, it can mislead the anomaly
detection efforts by other companies by including poorly calibrated or distorted readings. 

To defend against adversarial operators, we
propose an incentive mechanism that \textit{rewards} operators for each measurement, but \textit{penalizes} them for missing
and anomalous measurements. 
In particular, companies earn or lose \textit{credit} tokens based on the quantity and quality of their measurements. When a company runs out of credits, it is \textit{expelled} from the system. To avoid expulsion, a company must \textit{purchase credits} from other companies. This provides a financial incentive for companies to accumulate credits. We describe our incentive mechanism in the following.

\subsection{Description of the Mechanism}
\label{sec:mechanism}
In our mechanism, the grid
operators first agree on a topology -- i.e., on $\nodes$, $\nodes'$, and $\lines$. Each operator $o \in \owners$ then \textit{registers} a set of meters $\meters_0$ and states the position of each meter within the topology. The grid operators are each given $\initdeposit$ initial
\textit{credits}, where $\initdeposit$ is a large integer. In each time slot, the MDMS server of each
operator $o$ shares a single measurement from each meter $m \in \meters_o$ with the rest of the consortium. After the
cut-off period has passed, the slot is \textit{finalized} -- at this point, the measurements are evaluated and the credits are redistributed in the following
way. Each operator $o \in \owners$ receives a reward $\reward$ for each meter in $\meters_o$ that has shared a
measurement, and a penalty $\misspen$ for each meter in $\meters_o$ that has \textit{not} shared a
measurement. If at least one operator has not shared a measurement then the incentive mechanism is terminated after administering the missing measurement penalties.
However, if \textit{all} measurements have been received then the algorithm performs anomaly detection using state
estimation. In particular, it constructs the $(M+N') \times 1$ vector $\vec{x}$ as follows: for each $i=1\ldots,M$,
the $i$th element of $\vec{x}$ is set equal to the measurement of the $i$th meter. The next $N'$ elements are set equal to zero.
Finally, the residuals are computed as $\hat{\epsilon} = \vec{x} - P\vec{x}$, where $P$ is the projection
matrix of the weighted least squares estimator (as discussed in Section~\ref{sec:background}). We then compute the sum of the meters' squared residuals as $r = \sum_{m=1}^{M}
\hat{\epsilon}^2_m$. (We multiply the rows in the grid matrix $G$ that correspond to the zero-injection nodes by a large constant $\kappa$, so the $N'$ residuals that correspond to those nodes are typically negligible.) If $r$ is above an agreed threshold $\anomthres$, then the operators are penalized proportionally to the squared residuals of their meters. Let $\anompen$ be an agreed constant that determines the overall severity of the anomaly penalties. The penalties are then as follows: we iterate over all meters $m \in \meters$ and remove the following number of credits from of $m$'s operator:
\begin{equation}
\Phi'_m = \frac{\anompen}{r}\left(\epsilon^2_{m} - \frac{r}{M}\right).
\label{eq:penalty}
\end{equation}
That is, to compute  $\Phi'_m$, we subtract from $m$'s squared residual $\epsilon^2_{m}$ the fraction $\frac{r}{M}$, i.e., the system-wide average squared residual. The result is then multiplied by $\frac{\anompen}{r}$, which ensures that the maximum penalty per meter remains below $\anompen$. This prevents an extremely high residual, possibly due to an FDI attack or a fault, from causing an operator to lose all of its credits in a short time period. Note that $\Phi'_m$ can be either \textit{positive} or \textit{negative}, i.e., some operators profit if others are penalized.
In fact, it is easy to check that summing \eqref{eq:penalty} over all $m \in \meters$ yields $0$, so no credits are ever created or lost. Similarly, when a node
is rewarded by $\reward$ or penalized by $\misspen$, then each of the other $O-1$ operators loses or gains $\frac{\reward}{O-1}$ or $\frac{\misspen}{O-1}$
credits, respectively. As such, the total number of credits is fixed by design.

We note that if two operators $o$ and $u$ are symmetric in the sense that they have the same number of meters and a one-to-one mapping between any $m\in \meters_o$ and a $m' \in \meters_u$ such that $p_m = p_{m'}$, then their expected net gains/losses are also the same. Hence, if \emph{all} operators have exactly the same meters and error probabilities, then each of them will have a net loss of zero. In any other case, the credits of at least one operator will tend to zero over time. This means that smaller operators must periodically \textit{buy credits} from larger operators to avoid expulsion -- hence, larger operators are rewarded for contributing more meaningful data. This is an essential feature of our mechanism, as operators would otherwise not be incentivized to share as much data as they can. We assume that payments for credit transfers are done out-of-band, and
leave alternative means of achieving this, for example, using a two-way peg with a smart contract on another
blockchain such as Ethereum \cite{back2014enabling}, as future work. 
The speed at which companies run out of credits depends on the parameter choice. In Appendix~\ref{sec:example}, we present a realistic choice of parameters and illustrate our mechanism through a numerical experiment.

\subsection{Properties of the Incentive Mechanism}
\label{sec:incentive_properties}

In this section, we investigate the properties that determine our mechanism's incentive compatibility.
In particular, we present four corollaries that formalize the conditions under which our mechanism is incentive compatible. Each corollary considers a different adversarial strategy: not registering valid meters, withholding measurements from registered meters, blocking measurements from other operators' meters, and distorting measurements. These corollaries follow straightforwardly from three theorems that are stated and proven in Appendix~\ref{sec:theorems}.

 Each theorem focuses on the expected gains of an operator $o \in \owners$ in a single time slot, as changes to $o$'s credits are accrued over a succession of slots.
After each time slot is finalized, $o$'s credits are affected by each meter $m \in \meters_o$ in the following three ways: 
\begin{enumerate} 
\item $o$ earns a reward $\reward$ if $m$ produced a measurement, 
\item $o$ incurs a penalty $\misspen$ if $m$ did not produce one, and 
\item if there are no missing measurements and $r > \anomthres$, then $o$ earns or loses $\Phi'_m$ credits in accordance with \eqref{eq:penalty}. 
\end{enumerate}
Furthermore, $o$'s credits are affected by the meters of other operators. Due to the redistribution of the credits, $o$ loses or gains credits if other meters produce or miss measurements, respectively. Finally, $\Phi'_m$ depends on other meters' measurements through $\epsilon^2_j$ and $r$, and whether it is applied at all depends on whether any other meter failed to produce a measurement. In Appendix~\ref{sec:theorems}, we formalize this by decomposing $o$'s expected gains during a time slot into a summation with five terms:
$$
\rnet = \rtot - \phitot - \roth + \phioth + \phip.
$$
Each of the three theorems considers the impact on these terms of the different adversarial strategies. 

Our first corollary formalizes the conditions under which it is profitable to register a meter $m$, or to hide its existence from the other operators.

\begin{corollary}
It is profitable for $o$ to add any meter $m \in \meters_o$ if the impact of adding the meter $m$ on $\phip$ is negligible and $$
\missprob_m < \left(1 + \frac{\misspen}{\reward}\right)^{-1}.
$$
\label{cor:addmeters}
\end{corollary}

\begin{proof}
This follows from setting $\Delta_m(\iiadd)=0$ in Theorem~\ref{thm:addmeters} in Appendix~\ref{sec:theorems}.
\end{proof}

As it is profitable for $o$ to add suitable meters regardless of the actions of other operators, it is a \textit{Nash equilibrium} for all operators to add all meters $m$ for which $p_m$ satisfies the bound in Corollary~\ref{cor:addmeters}. In practice, changes to $\phip$ can be small because the maximum penalty $\anompen$ is low, because $\anomthres$ is high, or because the new meter has a small impact on the residuals. The bound in Corollary~\ref{cor:addmeters} depends only on the ratio $\frac{\misspen}{\reward}$ -- e.g., if $\reward$ and $\misspen$ are equal, then $\zeta = \frac{1}{2}$, which means that the meter must be online at least half the time for the rewards to outweigh the penalties. In practice, $\misspen$ should be high compared to $\reward$ (e.g., $\frac{\misspen}{\reward} > 1000$) to ensure that meters with a relatively high probability of being offline are not added.

The next corollary formalizes whether it is profitable for $o$ to share or withhold the readings produced by its meters.

\begin{corollary}
It is always profitable for $o$ to add the readings of all of its meters unless
$$
\reward + \misspen < \sum_{m \in \meters_o} \Phi'_m
$$
and $o$ is unable to block a reading by another operator's meter.
\label{cor:withhold}
\end{corollary}

\begin{proof}
From Theorem~\ref{thm:addmeasurements}  in Appendix~\ref{sec:theorems}, we know that $o$'s profit from withholding measurements equals $\delta + q {\bf 1}_{M'} \sum_{m \in \meters_o} \Phi'_m$, with $\delta$ and ${\bf 1}_{M'}$ as defined in the statement of the theorem. The operator $o$ maximizes $\delta$ by not withholding any of its own measurements, i.e., choosing $M'_o = 0$. However, if $o$ also does not block any other operator's measurement (i.e., $M_{\neg o}' = 0$), then ${\bf 1}_{M'} = 0$. In this case, choosing $M_o' = 1$ results in a profit of $-(\reward + \misspen) + \sum_{m \in \meters_o} \Phi'_m$. This is positive if the bound in the corollary holds. \end{proof}

Informally, the corollary states that it is profitable for $o$ to add all possible readings for its meters unless its expected loss of credits through the anomaly penalties outweighs the loss of $\reward$ and the penalty $\misspen$. In this case it would be profitable to withhold a single meter's measurement. In practice, this situation can be avoided by choosing a much smaller value for $\anompen$ than for $\misspen$ and $\reward$.

Although it is generally profitable for $o$ to add measurements, the zero-sum nature of credit redistributions implies that it is also profitable to block other operators' measurements. This is formalized through the next corollary.

\begin{corollary}
If it possible for $o$ to block $n$ distinct measurements of meters owned by other operators, then it is profitable to do so unless
$$
\frac{n}{O-1} (\reward + \misspen) < \sum_{m \in \meters_o} \Phi'_m.
$$
\label{cor:block}
\end{corollary}

\begin{proof}
Again, we obtain from Theorem~\ref{thm:addmeasurements} that $o$'s profit from withholding measurements equals $\delta + q {\bf 1}_{M'} \sum_{m \in \meters_o} \Phi'_m$. The operator $o$ maximizes $\delta$ by choosing $M'_o$ as low as possible and $M_{\neg o}'$ as high as possible. If $q {\bf 1}_{M'} = 0$ or $\sum_{m \in \meters_o} \Phi'_m \geq 0$ then it is always optimal to maximize $\delta$ by choosing $M'_{\neg o} = n$ and $M_o' = 0$. However, if $\sum_{m \in \meters_o} \Phi'_m < 0$, i.e., if $o$ is making a profit from the anomaly penalties, then it can be profitable for $o$ to choose $M'_{\neg o} = 0$ and $M_o' = 0$, but only if the bound in the corollary holds.
\end{proof}

The result of Corollary~3 is a necessary consequence of the incentive mechanism: if $o$ has no incentive to block measurements, then $o$ also has no incentive to share measurements. It should therefore be made as hard as possible for $o$ to block measurements. We discuss this in more detail in Section~\ref{sec:analysis}.

The final corollary formalizes the conditions under which it is profitable for $o$ to distort measurements. This is a technical condition that depends strongly on the structure of the grid and the location of $o$'s meters within it. In the following, let \mbox{$\distortmeters \subset \meters$} be the set of smart meters whose measurements can be arbitrarily distorted by $o$. Typically $\distortmeters = \meters_o$ but this is not a requirement. Let $d_{m}$, $m \in \meters$, be the distortion applied to $m$'s measurement, and $x_m$ the ``true'' -- i.e., corresponding to the physical reality at $m$'s claimed location in the grid -- measurement of meter $m$. We then construct the $(2M+N') \times 1$ vector $\vec{y}$ as follows: the first $M$ entries correspond to $d_1,\ldots,d_M$, the next $M$ values correspond to $x_1,\ldots,x_M$, and the final $N'$ entries correspond to the zero-injection nodes. 
We can then prove the following corollary.

\begin{corollary}
With $\vec{y}$ as defined above and matrix $A$ as defined in Appendix~\ref{sec:theorems}, it is profitable for $o$ to introduce distortions $d_m$, $m \in \distortmeters$, given $x_m$, $m \in \meters$, if and only if 
$
\;\vec{y}^{\,T} A \vec{y} < 0.
$
\label{cor:distort}
\end{corollary}
\begin{proof}
This is a reformulation of Theorem~\ref{thm:distortmeasurements}  in Appendix~\ref{sec:theorems}.
\end{proof}
In practice, an honest operator can determine whether a coalition of other operators is able to profit from distorted measurements by checking whether some combinations of distortions $d$ and measurements $x$ exist such that $\vec{y}^{\,T} A \vec{y} < 0$.
However, the operator should be careful to eliminate trivial attacks: e.g., where a single meter has an error, and the ``attack'' is to  correct this error. This can be achieved by introducing a suitable constraint: e.g., that the residuals $\epsilon = P\vec{x} -\vec{x}$ are equal to zero without the attack. This can be formulated as a Quadratic Programming (QP) problem:
\begin{equation}
\setlength{\tabcolsep}{0pt}
\text{
\begin{tabular}{lrlc}
minimize $\;\;\;$ & $\vec{y}^{\,T} A \vec{y}$ & &\\
subject to & $P\vec{x} - \vec{x}$ & $\; = 0$& \\
	& $x_m$ & $\; \geq 0$&, $m \in \meters$ \\
	& $d_m$ & $\; = 0$&, $m \notin \distortmeters$. \\
\end{tabular}
}
\setlength{\tabcolsep}{6pt}
\label{eq:qp}
\end{equation}
If the result found by solving \eqref{eq:qp} equals 0, then no profitable attack exists, whereas a profitable attack does exist if the result is negative. Solving a QP is computationally efficient if $A$ is positive semidefinite, but this is not always true in our setting. In this case, solving the QP is NP-Hard. Even in this case, state-of-the-art solvers such as BARON or IBM's CPLEX are typically able to solve QPs for several dozen variables. An alternative formulation of incentive compatibility  is to add the constraint $\vec{y}^{\,T} A \vec{y} < 0$ to \eqref{eq:qp}, and change the objective function to an arbitrary constant. Our mechanism is then incentive compatible if no feasible solution can be found. However, we have found that some of the software tools -- in particular CPLEX -- support matrices that are not positive semidefinite in the objective function, but not in the constraints. We leave a further investigation into the use of QP tools as future work.

\section{Implementation}
\label{sec:implementation}
In this section, we present an implementation of our data sharing solution, including the incentive mechanism described in the previous section.
In particular, we implement our solution on a private blockchain maintained by the consortium. The smart
contract processes data uploads and operates the incentive mechanism. Each operator runs a blockchain node on one of
its MDMS servers. We use Hyperledger Fabric for our implementation, due to its three advantages over the
alternatives, e.g., private Ethereum. First, it has higher throughput~\cite{dinh2017blockbench}. Second, it
supports a centralized Membership Service
Provider (MSP) that allows us to easily add and remove operators from the system.
Third, while Ethereum executes its contracts on its own virtual machine, Hyperledger supports native
code execution. In particular, a Hyperledger contract can be written in the Go language and import third-party
libraries, e.g., libraries that support fast matrix computations. 

\subsection{Data Model}
\label{sec:data_model}

The smart contract stores a variety of data fields, including information about the grid topology, the measurements, the system-wide constants and parameters, and information that is cached (e.g., the projection matrix $P$).
We note that Hyperledger only allows data to be stored in the form of key-value pairs, so during processing we occasionally need to parse strings into arrays or binaries. We did not find that this had a major impact on performance. An overview of the key-value pairs stored in the contract can be found in Table~\ref{tab:data}.

\begin{table}[h]
\centering
\caption{Key-value pairs stored in the smart contract.}
\label{tab:data}
\begin{footnotesize}
\setlength{\tabcolsep}{3pt}
\begin{tabular}{cccc}
\toprule
name & value & key & data type \\ \midrule
$n_1$ & 1st connected bus ID & line ID & integer-integer \\
$n_2$ & 2nd connected bus ID & line ID & integer-integer \\
$o$ & operator ID & meter ID & integer-integer \\
$X$ & reactance & line ID & integer-float \\
$l$ & line ID & meter ID & integer-integer \\ 
$n$ & node ID & meter ID & integer-integer \\ 
$s$ & type: $\arraycolsep=2pt\begin{array}{lcl} 0 & \Rightarrow & \text{SCADA bus} \\ 1 & \Rightarrow & \text{SCADA line} \\ 2 & \Rightarrow & \text{PMU} \end{array}\arraycolsep=5pt$ & meter ID & integer-integer \\
$d$ & deposit (credits) & operator ID & integer-integer \\ 
$\measurements'$ & recorded measurements & epoch & integer-float list \\ \bottomrule
\end{tabular}
\setlength{\tabcolsep}{6pt}
\end{footnotesize}
\end{table}

\subsection{Smart Contract Functions}
\label{sec:sc_functions}

\begin{algorithm}[h]
    \caption{Measurement processing.}
\label{alg:process_measurement}
\footnotesize
\SetKwProg{func}{function}{}{}

\func{changeDeposit($o$, $D$)}{
	$D' \gets \max(-d[o], D)$ \Comment cannot subtract more from $o$ than $d[o]$

	$\Delta^* \gets 0$
	
	$\owners' \gets \emptyset$ \Comment{operators with non-zero credit}
	
	$O \gets 0$
	
	\For{$o' \in \owners$}{ \label{ln:owners}
		\If{$d[o] > 0$}{
		 	$\owners' \gets \owners' \cup \{o\}$
			
			$O \gets O + 1$
		 }\vskip-0.12cm
	}\vskip-0.12cm
	
	\For{$o' \in \owners'$}{
		\If{$o' \neq o$}{
			$\Delta \gets \left\lfloor \frac{D'}{O-1} \right\rfloor$
			
			$d[o'] = d[o'] - \Delta)$
			
			$\Delta^* \gets \Delta^* + \Delta$
		}\vskip-0.12cm
	}\vskip-0.12cm
	
	$d[o] \gets d[o] + \Delta^*$ \Comment{\begin{tabular}{l}due to rounding, \\ it may be that $\Delta^* \neq d$\end{tabular}}
}
\vskip-0.06cm

\func{finalizeTimeSlot($t$)}{

	$\meters' \gets \emptyset$ \Comment{meters that reported a measurement}
	
	\For{$(v, m) \in \measurements'[t]$}{ \label{ln:meas-for}
		$x_{m} \gets v$
			
		\textit{changeDeposit($o[m]$, $\reward$)}
		
		$\meters' \gets \meters' \cup \{m\}$
	}\vskip-0.12cm
	
	\If{$\meters' = \meters$} {
	
		\For{$i\in\{1,\ldots,N'\}$}{ \label{ln:zero-for}
			$x_{i+M} = 0$
		}\vskip-0.12cm
		
		$\vec{x} \gets (x_1,\ldots,x_{M+N'})^{T}$
	
		$\vec{\epsilon} \gets \vec{x} - P \vec{x}$ \label{ln:mat-vec-mult}
	
		$r \gets \sum_{j=1}^{M} \epsilon^{2}_j$
		
		$\Delta^* \gets 0$ 
	
		\If{$r > \theta$} {
		    $k = r / M$
		
			\For{$j \in \{1\ldots,n\}$}{ \label{ln:anom-pen-for}
			    $\Delta \gets \lfloor \anompen(\epsilon^{2}_j - k)/r \rfloor$
			
			    $d[o[m_j]] \gets d[o[m_j]] - \Delta$ 
			    
				$\Delta^* \gets \Delta^* + \Delta$
			}\vskip-0.12cm
			
			$m^* \gets \text{RandomValue}(1,\ldots,M)$
			
			$d[m[o^*]] \gets d[m[o^*]] - \Delta^*$
			
		}\vskip-0.12cm
	}\vskip-0.12cm \Else{
		\For{$m' \in \meters \backslash \meters'$}{
			\textit{changeDeposit($o[m']$, $-\misspen$)}
		}\vskip-0.12cm
	}\vskip-0.12cm
}\vskip-0.06cm

\func{processMeasurement($v'$, $t'$, $m'$)}{
	$c \gets$ \textit{currentBlock()} 
	
	$\measurements'[t'] \gets \measurements'[t'] \cup \{(v', m')\}$

	\While{$c > (f+1) \cdot \slotduration + \timeoutblocks$}{ \label{ln:while-loop}
		\textit{finalizeTimeSlot(\hspace{1pt}f+1)}
		
		$f \gets f+1$ 
	}\vskip-0.12cm
}
\end{algorithm}

\begin{algorithm}[h]
    \caption{Updating the grid.}
\label{alg:grid_change}
\footnotesize
\SetKwProg{func}{function}{}{}
\func{findTopologyMatrix($\nodes$, $\nodes'$, $\lines$, $\meters$)}{
	\For{$i \in 1,\ldots,|\nodes|-1$}{ \label{ln:top_init_for}
		\For{$j \in 1,\ldots,|\meters|+|\nodes'|$}{
			$g_{i\,j} \gets 0$
		}\vskip-0.12cm
	}\vskip-0.12cm
	
	$i \gets 0$
	
	\For{$m \in \meters$} { \label{ln:top_m_for}
		$i \gets i+1$
		
		\uIf{$s[m] = 0$}{
			\For{$l' \in \lines$} { \label{ln:top_l_for}
				\If{$n_1[l'] = n[m]$}{ 
					$g_{i\,n_1[l']} \gets g_{i\,n_1[l']} - 1/X[l']$
					
					$g_{i\,n[m]} \gets 1/X[l']$
				}\vskip-0.12cm
			}\vskip-0.12cm
		} \uElseIf{$s[m] = 1$} {
			$g_{n_1[l[m]]} \gets -1/X[l[m]]$
			
			$g_{n_2[l[m]]} \gets 1/X[l[m]]$
		} \Else{
			$g[b[m]] \gets 1$
		}
		\vskip-0.12cm
	}
	\vskip-0.12cm
	
	\For{$n \in \nodes'$} { \label{ln:top_n_for}
		\For{$l' \in \lines$} {
			\If{$n_1[l'] = n[m]$}{
				$g_{i\,n_1[l']} \gets g_{i\,n_1[l']} - \largeconstant$
				
				$g_{i\,n[m]} \gets \largeconstant$
			}	
			\vskip-0.12cm
		}\vskip-0.12cm
	}
	\vskip-0.12cm
}
\vskip-0.06cm

\func{multiplyCheck($A$, $B$, $C$, $D$)}{
	$X = AB$
	
	$Y = XC$
	
	\Return $Y = D$ 
}

\func{updateGrid($\nodes$, $\nodes'$, $\lines$, $\meters$, $P'$, $U$)}{

	$G' \gets $ \textit{findTopologyMatrix($\nodes$, $\nodes'$, $\lines$, $\meters$)} 	
	\If{multiplyCheck($G'^{T}, G', U, I_{|\meters|+|\nodes'|}$)} { \label{ln:ggu} 
		\If{multiplyCheck($G', U, G'^{T}, P'$)}{ \label{ln:gug}
			$G \gets G'$
			
			$P  \gets P'$
		}
		\vskip-0.12cm
	}
	\vskip-0.12cm
}
\vskip-0.12cm

\end{algorithm}

\paragraph{Initialization}

As mentioned before, the consortium members agree before the initialization of the smart contract on the system parameters. 
These parameters are set to the agreed values when the contract is created. Furthermore, an initial grid topology \mbox{($\nodes$, $\nodes'$, $\lines$, $\meters$)} is chosen and the initial consortium members $\owners$ are given an initial deposit $\initdeposit$. 

\paragraph{Measurement Processing}

During each time slot, each node uploads one measurement per meter by calling the \textit{processMeasurement} function of Algorithm~\ref{alg:process_measurement}. This function adds the measurement to $\measurements'[t']$, which is the set of measurements for the current time slot $t'$, and checks for which previous time slots the cut-off period has passed. For each such previous time slot $t$, the effect on the deposits can be processed through a call to \textit{finalizeTimeSlot}. This function first evaluates which meters have reported a measurement during $t$ -- these measurements are added to the set $\meters'$ and the operator of the meter is rewarded through a call to \textit{changeDeposit}. If some meters have not reported a measurement, then they are punished and the routine terminates. However, if all nodes have reported a measurement then the anomaly detection routine is started. In particular, the residual vector $\hat{\epsilon}$ is determined, and if the sum of squared residuals $r$ exceeds the threshold $\anomthres$, then the deposits are changed in accordance with \eqref{eq:penalty}. 

We note that the function \textit{changeDeposit} has built-in checks to ensure that nodes with zero credits do not gain credits from the redistribution of penalties to other nodes. Furthermore, it checks that no credits are lost due to rounding errors. Rounding errors may also occur when administering the anomaly penalties: in this case, any remaining credits are given to the operator of a meter sampled uniformly at random from $\meters$.

\paragraph{Grid Change}

To change the grid -- e.g., to add/remove meters, nodes, or lines -- one node proposes the change by using a transaction to call the \textit{updateGrid} described in Algorithm~\ref{alg:grid_change}. This transaction lists the new values for $\nodes$, $\nodes'$, $\lines$, and $\meters$. The grid topology is calculated using these lists in the \textit{findTopologyMatrix} function of Algorithm~\ref{alg:grid_change}.
This function constructs the grid topology as discussed in \cite{wood2013power,li2015location} and particularly \cite{bi2014graphical}.
Additionally, the new matrices $U = (G^{T} G)^{-1}$ and $P$ are sent as parameters of Algorithm~\ref{alg:grid_change}. As we discuss in Section~\ref{sec:complexity_analysis}, this speeds up the function by moving a computationally expensive matrix inversion from the blockchain (which is performed by all nodes) to a single node. The matrices are checked for correctness in the smart contract through the function $\textit{multiplyCheck}(A,B,C,D)$ which returns \textsc{true} if $ABC=D$ and \textsc{false} otherwise. In particular, the algorithms checks whether $G^{T}G U = I_n$, where $I_n$ is the identity matrix with dimension $n$, and whether $G U G^{T} = P$. The interaction between the processes executed by a node's client and the blockchain is depicted in Figure~\ref{fig:flowchart}.

\begin{figure}[h]
\centering
\includegraphics[width=0.47\textwidth]{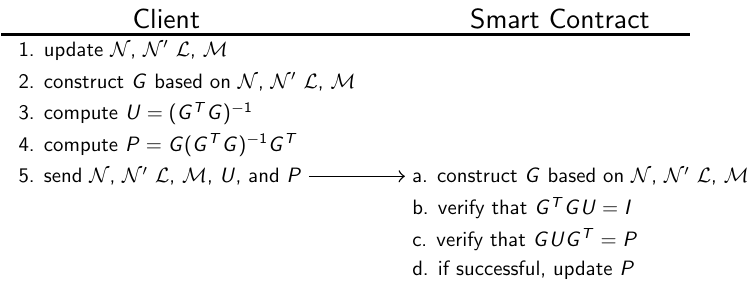}
\caption{Flowchart that depicts the processes that occur when the grid topology is changed.}
\label{fig:flowchart}
\end{figure}

Although we omit this from the pseudocode for brevity, all operators whose meters and added/removed must vote to approve changes to the grid. Other nodes do not need to approve the changes, but have the ability to cast a veto within a cooldown period if the change is suspected to be the work of an attacker -- e.g., removing all or nearly all of an operator's meters, or adding too many meters.

\paragraph{Other Functions}

We briefly discuss the remaining smart contract functions, but we do not present pseudocode for them because they are straightforward to implement. First, operators will occasionally transfer funds between each other, particularly if one has almost depleted its credits. This can be done via a call to a \textit{transferCredits} function that removes credits from the sender's deposit and adds the same amount to the recipient's deposit.

To expel an operator who has run out of credits, a function \textit{expelOperator} can be called to remove the operator from $\owners$, and remove all meters that it operates. To add a new operator, a function \textit{addOperator} can be called by an existing operator $o$ -- as part of this function, $o$ must also send credits to the new operator. Otherwise, the addition/removal of the node (and the suspension of read access) is processed via the MSP in Hyperledger Fabric.

Finally, the operators may occasionally want to change the system parameters. This could be for a variety of reasons, e.g., when the anomaly detection threshold $\anomthres$ is found to be too high or low relative to the natural variance in the measurements. In this case, one node proposes a system parameter change through a call to the \textit{proposeParameter} function. In this case, a supermajority of operator must vote to approve the change -- note that the supermajority threshold itself can be a modifiable system parameter.

\subsection{Asymptotic Complexity}
\label{sec:complexity_analysis}

In this section we analyze the asymptotic complexity of the functions in Algorithms~\ref{alg:process_measurement}~and~\ref{alg:grid_change}, and identify the bottlenecks that we seek to ameliorate. Before we begin, we note that the time complexity of multiplying an $N \times M$ matrix by an $M \times K$ matrix is $\Theta(NMK)$.\footnote{While some practical approaches exist that achieve slightly higher speeds, e.g.,  Strassen's algorithm, this does not have an impact on our analysis.} Furthermore, the time complexity of checking the equality of two $N \times M$ matrices is $\Theta(NM)$. An overview of the time complexities of the most computationally intensive lines in the smart contract functions is given in Table~\ref{tab:complexities}. To obtain these results, we made two assumptions about the grid, in particular that $\NM \geq N - 1$ (as otherwise the matrix $G^{T}G$ would not be invertible), and that $\Theta(ML) = \Theta(MN) = \Theta(\NM N)$ (as there will typically be only several lines per node, and one or two meters per line).

We discuss the \textit{updateGrid} function in Algorithm~\ref{alg:grid_change} in detail, as it includes an important optimization step. The function starts by calling the \textit{findTopologyMatrix} function to obtain $G$, which has complexity \mbox{$\Theta(\NM N)$}. It then calls \textit{multiplyCheck} in line~\ref{ln:ggu}, which computes $G^{T} G$. This consists of multiplying an \mbox{$(N-1) \times \NM$} by an $\NM \times (N-1)$ matrix. The time complexity of this is $\Theta(\NM N^2)$. Next, $G^{T} G$ is multiplied by $U$, which has complexity $\Theta(N^3)$. It then checks equality with a $(N-1) \times (N-1)$ identity matrix, which requires $\Theta(N^2)$ operations. In line~\ref{ln:gug}, \textit{multiplyCheck} is invoked again. First, $G U G^{T}$ is computed, which requires the multiplication of a $\NM \times (N-1)$ matrix with a $(N-1) \times (N-1)$ matrix. This is followed by the multiplication of a $\NM \times (N-1)$ matrix with a $(N-1) \times \NM$ matrix. The first step has complexity $\Theta(\NM N^2)$ and the second step a complexity of $\Theta(\NM^2 N)$. 

If, instead of requiring $U$ and $P$ as transaction input, the smart contract would determine $P$ directly from $G$, then it would have to compute $G (G^{T}G)^{-1} G^{T}$ by itself. This would require three matrix-matrix multiplications, in contrast to the four used in the current implementation. However, it would also require the inversion of an $(N-1) \times (N-1)$ matrix. Inversion of an $(N-1) \times (N-1)$ matrix has the same theoretical time complexity as multiplication of two $(N-1) \times (N-1)$ matrices, but in practice it is much slower. Our approach is therefore faster than the na\"ive approach.

\begin{table}[htp]
\caption{Time complexity of the smart contract functions}
\begin{center}
\vskip-0.2cm
\begin{tabular}{cc}
function  & complexity \\ \toprule 
\textit{finalizeTimeSlot}  & $\Theta(\NM^2)$ \\[0.05cm]   
\textit{findTopologyMatrix} & $\Theta(\NM N)$ \\[0.05cm]   
\textit{updateGrid}  & $\Theta(\NM^2 N)$ \\
\bottomrule
\end{tabular}
\end{center}
\label{tab:complexities}
\vskip-0.5cm
\end{table}

\subsection{Performance Optimizations}
\label{sec:optimizations}

To detect an attack as soon as possible, it is preferable for the time slot duration to be as short as possible. Short time slot durations have the additional advantage of minimizing the variance between measurements taken at different time points in the slot. However, since the \textit{finalizeTimeSlot} must be executed at least once per time slot, it is essential that the computational overhead of this function is minimal. Furthermore, even though \textit{updateGrid} is only called sporadically, its computational overhead should not be such that it significantly impedes the processing of measurements. It is therefore essential that the performance of these functions is optimized.

To achieve this, we have implemented four practical optimizations. The first optimization is the offline computation and caching of $P$ as discussed in Section~\ref{sec:sc_functions}. By moving the most expensive part of the computation offline, we avoid interference with the processing of measurements. 

For the second optimization, we replace the function \textit{multiplyCheck}, which uses two expensive matrix-matrix multiplications, with the \textit{freivaldsCheck} function of Algorithm~\ref{alg:freivalds}. The latter uses Freivalds' algorithm \cite{freivalds1979fast,tramer2018slalom}, which is a randomized algorithm that requires four inexpensive matrix-vector multiplications instead of two matrix-matrix multiplications. Our implementation makes use of a cryptographic hash function $H_p$ that takes a byte array as input and returns an integer that is sampled uniformly at random from $\{0,\ldots,p-1\}$. We then construct a random vector $\vec{r}$ by repeatedly applying $H_p$ to a concatenation $A|B|C|D$ of the four input variables. The function then evaluates whether $ABC\vec{r} = D\vec{r}$. It can be shown that if $ABC \neq D$, then the probability that \textit{freivaldsCheck} returns \textsc{true} is at most $\frac{1}{p}$, which is negligible for large $p$.

\begin{algorithm}
    \caption{Fast verification of matrix equalities.}
\label{alg:freivalds}
\footnotesize
\SetKwProg{func}{function}{}{}
\func{freivaldsCheck($A$, $B$, $C$, $D$)}{
	
	$\vec{r} = (r_1\ldots,r_{m_D})$
	
	$r_1 \gets H_p(A|B|C|D)$
	
	\For{$i \in \{2,\ldots,m_D\}$} {
		$r_i \gets H_p(r_{i-1})$
	}
	\vskip-0.12cm

	$\vec{c} \gets C \vec{r}$ \label{ln: cr}
	
	$\vec{b} \gets B \vec{c}$ \label{ln: bc}
	
	$\vec{a} \gets A \vec{b}$ \label{ln: ab}

	$\vec{d} \gets D \vec{r}$
	
	\Return{$\vec{a} = \vec{d}$}
}
\vskip-0.06cm
\end{algorithm}

The third is the efficient storage of measurement data by partitioning the set $\measurements'$ into different sets $\measurements'[t]$ for each time slot $t$. We also store the elements of these sets using a different key-value pair for each measurement, instead of using a single key-value pair to store the entire set. The former approach reduces the complexity of write operations, since it avoids reading the existing entry, updating it, and writing the new combined value to the storage -- also known as {\em write amplification}. By contrast, the latter approach would reduce the number of read operations in \textit{finalizeTimeSlot}. However, we found that the write operations were a greater bottleneck than the read operations, as we discuss further in Section~\ref{sec:experiments}.

Finally, we found in Section~\ref{sec:complexity_analysis} that the bottlenecks in both \textit{finalizeTimeSlot} and \textit{updateGrid} were matrix-vector or matrix-matrix multiplications. Both are trivially paralellized, since the computation of the elements in different rows can be done independently.

\section{Analysis}
\label{sec:analysis}
In this section, we analyze the data sharing solution presented in Sections~\ref{sec:incentives}~and~\ref{sec:implementation} in terms of the requirements of Section~\ref{sec:requirements}, i.e., \textit{FDI attack security}, \textit{incentive compatibility}, and \textit{scalability}. For each requirement, we discuss the technical conditions under which it is satisfied, and we conclude the section by discussing the practical relevance of these conditions and our modeling assumptions.

\subsection{Requirements}
\label{sec:requirements_analysis}

\textit{FDI Attack Security: } for this requirement to be satisfied, FDI attacks against a compromised operator must still be detectable if redundant measurements are performed by uncompromised operators.
To make this formal, we consider an attack in which a set of operators $A \subset \owners$ is compromised. In such an attack, all meters \mbox{$m \in \meters'(A)$} where $\meters'(A) = \cup_{o \in A} \meters_o$ are insecure and their measurements can be arbitrarily distorted by the FDI attacker, whereas the other meters are secure and their measurements cannot be distorted. This scenario is well-known from the FDI literature \cite{bobba2010detecting,chen2006placement} -- e.g., \cite{chen2006placement} discusses placement strategies for secure meters that ensure that the system remains observable even if the insecure meters are subject to an FDI attack. In our setting, each operator $o \in A$ should consider the system that consists only of the meters in $\meters \backslash \meters'(A)$ and the nodes in $\nodes$ that are \textit{sensitive} to an FDI attack -- i.e., those nodes for which it holds that if their state is not estimated correctly, then grid operation may fail, e.g., leading to blackouts. If this system is observable (i.e., the matrix $G$ has full rank), then the data sharing solution satisfies the property of FDI attack security.

\textit{Incentive Compatibility -- Joining the Network: } the first part of our incentive compatibility requirement is satisfied if operators are incentivized to join the data sharing consortium. This holds for any operator $o$ if $o$'s benefits of joining outweigh the costs. To determine whether this is true for $o$, we first determine the expected credit change per time slot $\rnet$ as defined in Section~\ref{sec:incentives}. A first observation is that if $\rnet$ is positive, then it is \textit{always} profitable for $o$ to join the network. However, if $\rnet$ is negative, then $\rnet$ needs to be outweighed by the reduced \textit{expected costs of an FDI attack} per time slot, which we denote by $K$. 
To compute $K$, the operator must determine for each attack scenario $A \subset \owners$ the quantities $p^{\textsc{ATT}}_{A}$, ${\bf 1}^{\textsc{DET}}_A$, $K_A$, and $K'_A$, which are defined as follows:

\begin{center}
\vspace{0.1cm}
\begin{tabular}{lp{0.8\linewidth}}
$p^{\textsc{ATT}}_{A}$ & the likelihood that attack scenario $A$ occurs in a single time slot (e..g, if $A$ is expected to occur once a year, and a year consists of $S$ slots, then $p^{\textsc{ATT}}_{A} = S^{-1}$), \\
${\bf 1}^{\textsc{DET}}_A$ & $1$ if the attack $A$ can be detected by the meters in $\meters \backslash \meters'(A)$ and $0$ otherwise,\\
$K_A$ & the expected cost of attack $A$ if it is detected, and \\
$K'_A$ & the expected cost of attack $A$ if it \textit{not} detected.\\ 
\end{tabular}
\vspace{0.1cm}
\end{center}
Given these quantities, $o$ is incentivized to join the network if
$$
-\rnet < \sum_{A \subset \owners} p^{\textsc{ATT}}_{A} {\bf 1}^{\textsc{DET}}_A (K_A - K'_A),
$$
that is, if the expected costs per time unit of participating in the mechanism are smaller than the expected cost reduction per time unit from attacks that the mechanism can detect.

\textit{Incentive Compatibility -- Adversarial Behavior: } the conditions under which each of the five types of adversarial behavior presented in Section~\ref{sec:threat_model} are disincentivized are as follows.

\begin{itemize}
\item \textit{Refraining from registering valid meters}: from Corollary~\ref{cor:addmeters}, we know that it is always profitable for $o$ to add each meter~$m$ whose probability of missing a measurement $p_m$ is low enough, i.e., below the bound stated in the corollary. As such, it is only profitable to add those meters for which the probability of missing a measurements is low enough.
\item \textit{Uploading fake measurements from offline or non-existing meters:} we know from Corollary~\ref{cor:withhold} that is always profitable to upload as many measurements as possible, even from fake or offline meters. However, from Corollary~\ref{cor:distort} we know that sharing measurements that do not reflect the true system state can be costly, depending on the grid topology.
\item \textit{Taking its own meters offline:} we know from Corollary~\ref{cor:withhold} that this is never profitable.
\item \textit{Blocking other operators' measurements:} we know from Corollary~\ref{cor:block} that this is always profitable if it is possible.
\item \textit{Uploading distorted measurements:} we know from Corollary~\ref{cor:distort} that this can be costly, depending on the grid topology.
\end{itemize}

\textit{Scalability: } as discussed in Section~\ref{sec:optimizations}, we have presented four optimizations that improve our data sharing solution's practical performance: offline computation of the matrix $P$, Freivalds' algorithm, partitioning the measurement dataset, and multithreading. To demonstrate that this is sufficient to support thousands of meters, we present an experimental evaluation in Section~\ref{sec:experiments}.

\subsection{Discussion: Practical Relevance}
\label{sec:discussion}

We now discuss the practical impact of the technical conditions mentioned in the previous section. We first discuss the impact of an operator's size on its expected gains and losses, and then discuss how to defend against types of misbehavior that are not disincentivized (e.g., blocking or distorting measurements).

\textit{Large vs.\ Small Operators: } in the previous section, we saw that operators join the consortium only if the expected costs of lost credits are outweighed by the expected cost of an FDI attack and whether such an attack can be prevented by the data sharing solution. The latter only holds if the operator's meters are redundant in a network that also includes the meters of other operators. In general, this is more likely to be the case for smaller operators than for larger operators. For example, in a system with regional operators (e.g., Japan), only the meters on the edge of the operator's region are likely to have redundant measurements performed by other operators. If an operator's region is smaller, then the fraction of its meters that are on the region's edge will typically be larger, as the region's interior will be smaller. Similarly, in a system with a handful or large TSOs and a multitude of smaller DSOs (e.g., Germany), the measurements that are most important to the DSOs (i.e., on lines with the largest power flows) will be performed on lines that are also of interest to the TSOs.

 Conversely, larger operators are in general more likely to profit from the credit redistribution mechanism than smaller operators, because they share a larger number of measurements. As such, our mechanism maximizes operator participation by incentivizing both smaller and larger operators to join: the larger operators because they gain credits, and the smaller ones because they gain a larger degree of protection from FDI attacks.

\textit{Alternative Defenses Against Persistently Adversarial Operators: } We saw in Section~\ref{sec:requirements_analysis} that the following three types of adversarial behavior are not always disincentivized: uploading fake measurements, blocking measurements, and (depending on the grid topology) distorting measurements. In each case, defense measures exist that make this type of adversarial behavior hard to execute. In the first case, uploading fake measurements can be made harder by requiring that measurements are \textit{signed} by the meter using a private key that is known only to the equipment manufacturer. Such functionality is supported by default in meters that are equipped with a Trusted Platform Module \cite{tcglibrary}.
In the second case, mechanisms that defend against transaction blocking attacks -- e.g., DoS or eclipse attacks -- are known from the blockchain security literature \cite{homoliak2020security}. In Appendix~\ref{sec:blockchain_security}, we discuss how to defend against such attacks in our setting using round-robin consensus. Finally, operators can use audits by external parties to detect measurement distortion attacks.

In practice, an operator would not only consider direct financial incentives before adopting an adversarial strategy, but also reputation damage and lawsuits. Some types of misbehavior (e.g., not adding a measurement) are legal but undesirable, whereas others are illegal (e.g., blocking another operator's measurements or distorting measurements). Depending on the political situation in the operator's country, a regulator could simply make all data sharing mandatory. However, as witnessed by the Tennet lawsuit \cite{eutennet}, but also the recent Volkswagen emissions and \mbox{LIBOR} fraud cases, laws alone cannot provide a strict guarantee against misbehavior that is hard to detect. As such, our mechanism disincentivizes mild and hard-to-detect cases of misbehavior (e.g., neglecting to add meters), while leaving only the more serious cases that are susceptible to detection by other operators to the regulator.

\section{Experiments}
\label{sec:experiments}
In this section we present our evaluation of the effect of our optimizations on performance. 
We use Hyperledger Fabric v0.6 with version 1.13 of Go for the smart contract (`chaincode' in Hyperledger terminology). All the transactions are signed using the SHA256 hash algorithm. For matrix multiplication and  serialization, we use the Go library \texttt{gonum-v1}. We run the Hyperledger nodes on top of 4 powerful VMs -- each VM runs \texttt{ubuntu-18.04}, and has 8 cores and 32GB RAM. For each Hyperledger node, we increase the Docker memory limit  to 25GB. In the following, we use the shorthand notation $K = M + N'$ for the dimension (``\textit{size}'') of the matrices that we use for benchmarking.

\begin{figure}
     \centering
     \subfloat[][Finalize and detect anomaly]{\includegraphics[width=.34\textwidth]{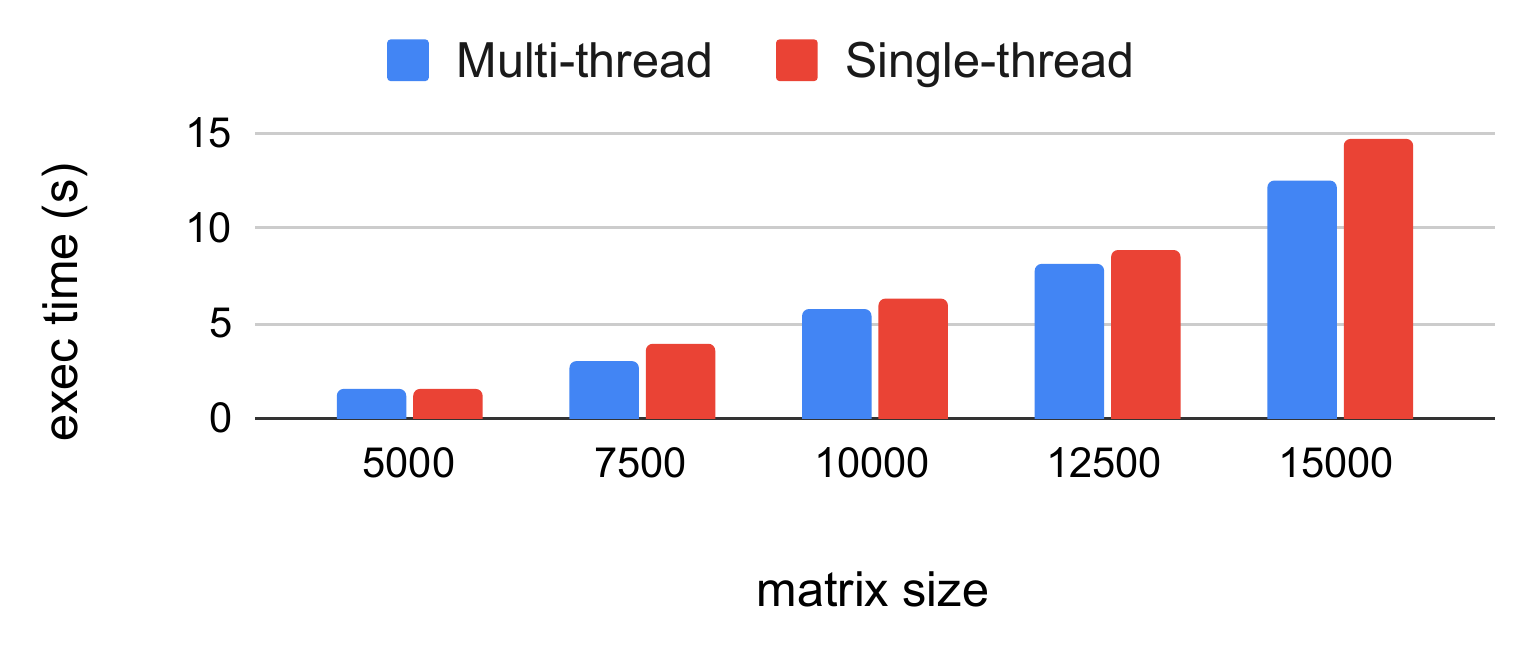}  
     \label{fig:plot1}}
     
     \vspace{-0.2cm}
     \subfloat[][Update projection matrix]{\includegraphics[width=.34\textwidth]{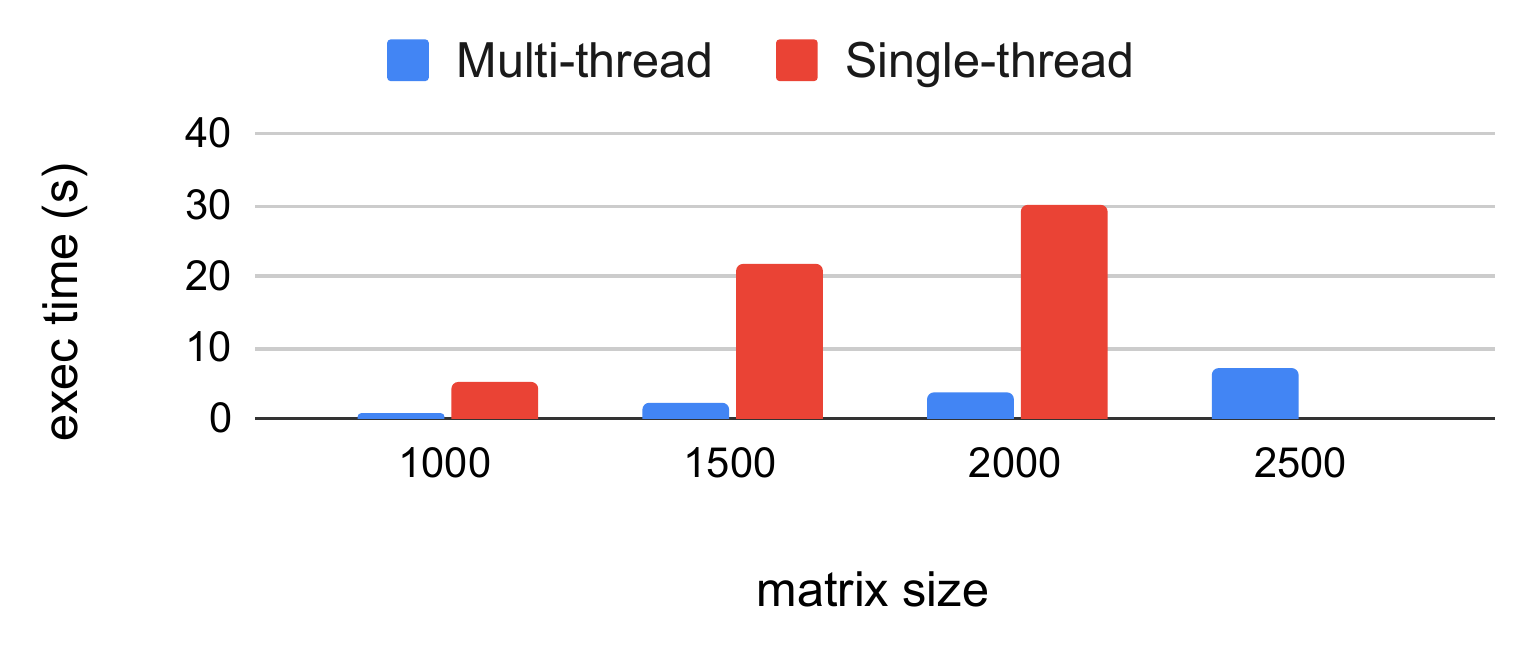}
     \label{fig:plot2}}
     \caption{Speed up due to parallel matrix multiplication.}
     \label{fig:latency}
     \vspace{-0.2cm}
\end{figure}

First, we evaluate the cost of grid updates and time slot finalizations for both single- and multi-threaded processing, to show the benefits of parallelizing matrix computations.
Figure~\ref{fig:plot1} shows the transaction execution time for finalization. Although loading the matrices from storage takes a large amount of time for large matrices, e.g., 2 seconds even for $K=5000$, we observe the benefits of parallel computation as the matrix size grows. In Figure~\ref{fig:plot2}, we display the execution times of updates to the projection matrix $P$, which is a $K \times K$ square matrix. The transaction size for $K=1000$ is 32MB and for $K=2000$ it is 128MB. Single-threaded computation takes a long time and even crashes at $K=2500$.

\begin{figure}
    \centering
    \subfloat[][Single vs multiple key-value pairs]{\includegraphics[width=.34\textwidth]{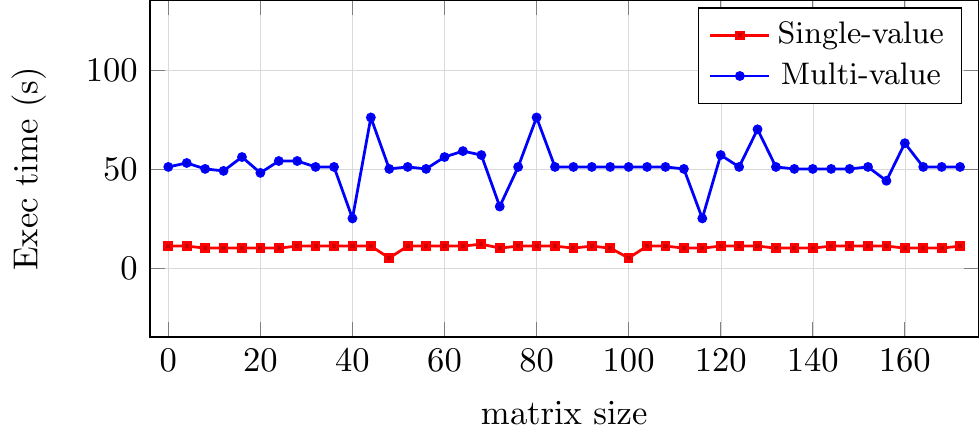}  
     \label{fig:plot3}}

\vspace{-0.2cm}
     \subfloat[][Single vs multiple thread anomaly detection]{\includegraphics[width=.34\textwidth]{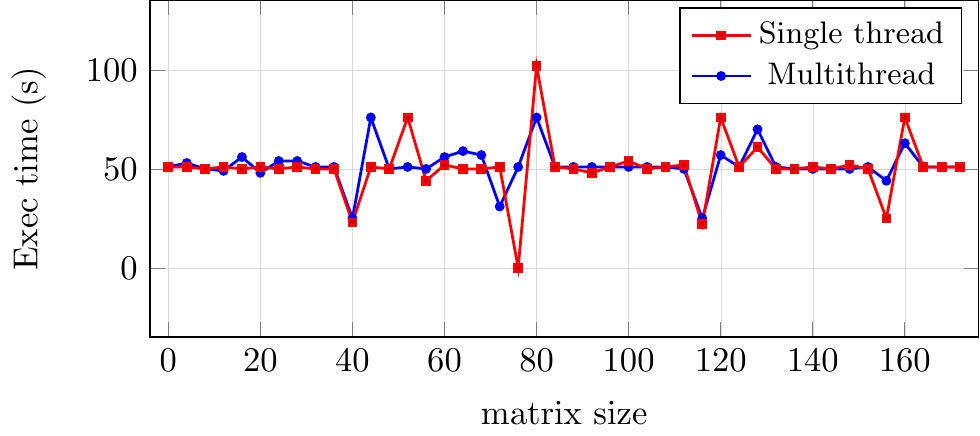}  
     \label{fig:plot4}}
     
      \caption{Throughput with different optimizations.}
      \label{fig:plot34}
      \vspace{-0.2cm}
\end{figure}

Next, we show the impact of our data storage and parallel execution optimizations on the throughput over time in Figures~\ref{fig:plot3}~and~\ref{fig:plot4}. Measurements are simulated and sent as transactions, and the time slots are finalized after 40-second intervals. We obtain the maximum throughput by increasing the requests per second for as long as the blockchain can manage. The matrix size $K$ equals $2000$ and the transaction size is 686 Bytes, including the signature. We clearly see that the throughput drops after each time slot finalization (at each 40 second mark) in both figures.

Figure~\ref{fig:plot3} shows the effect of different data layouts. When
we store multiple values on the same key, we suffer from write amplification (as discussed in Section~\ref{sec:optimizations}), so throughput is
low. Figure~\ref{fig:plot4} shows the impact of parallel execution on throughput. We see that the impact is
limited, because the execution time is too small to affect throughput. The \textit{average} throughputs for $K=2000$ are as follows: 51.6 for multi-threaded execution and 50.3 for single-threaded execution, which is not a major difference. We therefore find that although multithreading greatly accelerates grid updates, it does not have a major effect on time slot finalizations. However, data storage optimization does have a significant effect on the measurement transaction throughput.
We note that even though the throughput is low at 50 tx/sec, this means that we can set the time slot duration to 40 seconds for 2000 meters and hence collect measurements from each meter every 40 seconds.

\begin{figure}
    \centering
    \subfloat[][Single vs multiple key-value pairs]{\includegraphics[width=.34\textwidth]{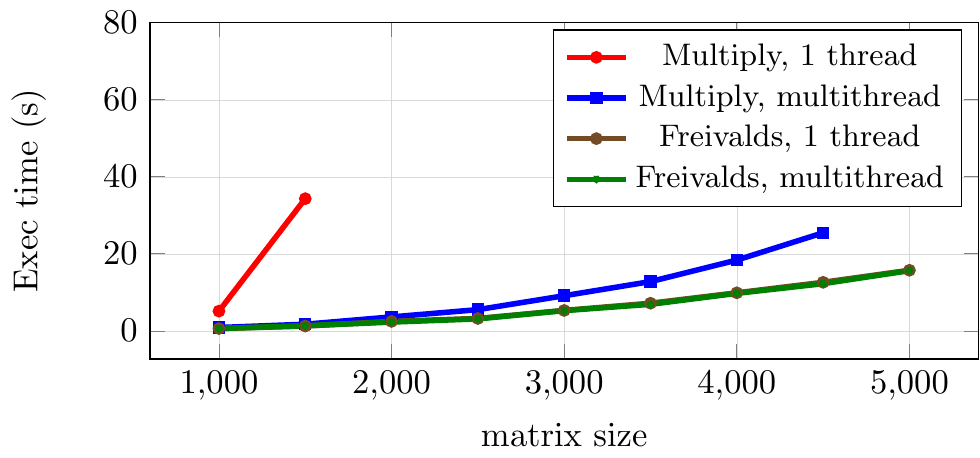}  
     \label{fig:plot5}}

\vspace{-0.2cm}
     \subfloat[][Single vs multiple thread anomaly detection]{\includegraphics[width=.34\textwidth]{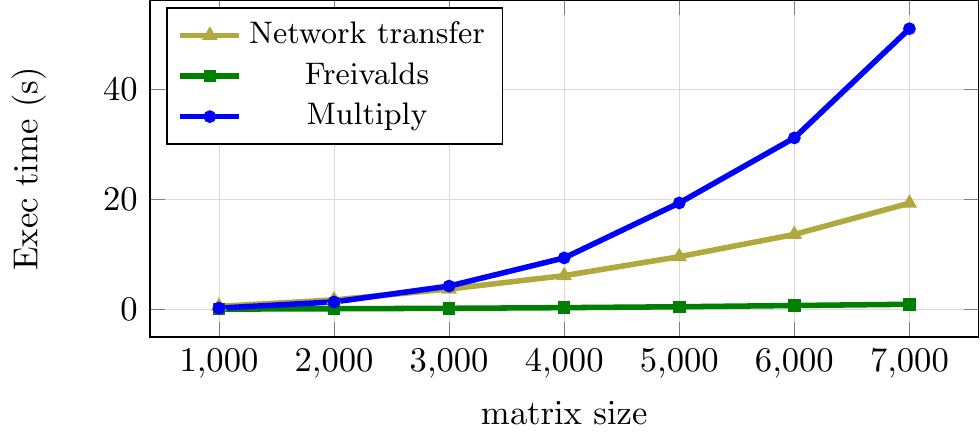}  
     \label{fig:plot6}}
     
      \caption{Effect of different performance optimizations on the execution time of the grid change algorithm.}
      \label{fig:plot56}
      \vspace{-0.2cm}
\end{figure}

Finally, Figure~\ref{fig:plot5} shows the reduction in execution time that we achieve by replacing the \textit{multiplyCheck} function of Algorithm~\ref{alg:grid_change} with the more efficient \textit{freivaldsCheck} from Algorithm~\ref{alg:freivalds}. We observe that the execution time drops further, from over 25 seconds for $K=5000$ to slightly over 12 seconds. However, we also observe little difference between executing the three matrix-vector multiplications sequentially in a single thread, (brown line) or in parallel (green line). To explain this, we display in Figure~\ref{fig:plot6} a breakdown of the cost into the network transfer cost -- i.e., transmitting the data of the fairly large $P$ and $U$ matrices -- and the computational cost of the two algorithms. In contrast to the \textit{multiplyCheck}, the cost of executing the \textit{freivaldsCheck} function is negligible compared to the network cost. As such, the potential to further reduce the execution time of this function through computational speed-ups is negligible.

\section{Related Work}
\label{sec:related}
Since we have already compared the paper to related work from the state estimation literature in Section~\ref{sec:model}, we instead focus in this section on comparing our work to related blockchain approaches for smart grids. 
Fan and Zhang \cite{fan2019consortium} propose a consortium blockchain that allows operators and consumers to exchange data, but they focus solely on privacy concerns that they address through encryption. In another paper by the same authors \cite{zhang2018blockchain}, a blockchain mechanism is proposed that allows devices such as smart meters to exchange fault diagnosis reports through a blockchain. Li et al.\ \cite{li2017consortium} present a blockchain framework that allows prosumers (consumers that produce electricity, e.g., via solar panels) to trade energy credits with other consumers. In general, blockchain solutions that allow users to exchange credits in a privacy-preserving manner are an active research field -- see also, e.g., \cite{gai2019privacy,mollah2020blockchain}.
These works are orthogonal to ours as they do not consider data sharing incentives, as the incentive to participate in their frameworks is to sell energy credits. A detailed overview of blockchain solutions applied to different  smart  energy  systems and applications can be found in \cite{hassan2019blockchain}.

\section{Conclusions \& Discussion}
\label{sec:conclusion}
In this paper, we have presented a blockchain solution that incentivizes operators to share data to detect FDI attacks. Our
solution provides security in settings where an individual operator is not able to detect attacks due to a
lack of redundancy. We have presented a formal analysis of our incentive mechanism, and shown that
operators are motivated to share data of as many meters as possible as long as the meter's error probability is
sufficiently low. Finally, we have presented four optimizations that allow our solution to be applied to
realistic grids with thousands of meters. 

Although the blockchain nodes may introduce a new attack vector, we note that an FDI attacker would not be able to violate network safety or liveness unless they compromise more than a third of all operators in the consortium. However, we do require that the blockchain client run by the operator is trusted, e.g., if the client does not validate the signatures of new blocks then our data sharing solution would be ineffective because the FDI attacker would be able to arbitrarily distort the measurements of other operators. However, the security of blockchain clients is a generic challenge that is not specific to smart grids.  
Bugs in the smart contract may also introduce a new type of vulnerability, but the smart contract's code can be  audited by all consortium members. 

In future work, we aim to investigate whether we can replace the MSP in Hyperledger, which we use to expel nodes, with a dedicated encryption mechanism that expels
nodes by resetting the encryption/decryption key among the remaining operators. Another interesting question is whether the asynchronous protocol in Hyperledger can be replaced with a more efficient bespoke algorithm for our context, e.g., using Byzantine broadcast \cite{cachin2011introduction} instead of distributed consensus. 
Finally, although we have focused on power grids, our results generalize to any other system in which multiple operators perform measurements that can be described using a system of linear equations, e.g., water or gas distribution networks.

\section*{Acknowledgments}
This research / project is supported by the National Research Foundation, Singapore, under its National Satellite of Excellence Programme ``Design Science and Technology for Secure Critical Infrastructure'' (Award Number: NSoE\_DeST-SCI2019-0009). Any opinions, findings and conclusions or recommendations expressed in this material are those of the author(s) and do not reflect the views of National Research Foundation, Singapore.
We also thank the anonymous reviewers of previous versions of this work for their insightful comments.

\bibliographystyle{abbrv}
\bibliography{ref}

\appendix
\label{sec:appendix}

\section{Construction of the Grid Topology Matrix}
\label{sec:grid_matrix}

In this section, we present more details about the construction of the grid topology matrix $G$.
The topology is represented using the bus-branch model \cite{monticelli1999state,abur2004power} that is commonly used in state estimation, which is the method that we use for anomaly detection. We only present a high-level overview of the model, and refer the interested reader to other works \cite{abur2004power,gomez2018electric,wood2013power,monticelli1999state} for a more in-depth discussion of power systems.

In our setting, buses represent nodes (e.g., power plants or transformer substations) whereas branches represent power lines. The sets $\owners$, $\nodes$, $\lines$, and $\meters$ are as defined in Section~\ref{sec:model}.
Recall that there are $N+1$ nodes/buses. Buses are represented by integers (via their IDs), so $\nodes \subset \N$.
There are $L$ lines/branches that connect the buses -- each line $l \in \lines$ can also be represented as a pair $(n,n')$ of buses (so by a pair of integers). Hence, $\lines \subset \N^2$.
In this definition, the lines are bidirectional: i.e., if $(n,n') \in \lines$, then $(n',n) \in \lines$ as well.
Let $\nb(n)$ be the set of \emph{neighbors} of $n$, i.e., 
$
\nb(n) = \{m \in \nodes : (n, m) \in \lines \}
$ 
The \emph{reactance} on line $l = (n,n')$ is denoted by $X_l = X_{nn'}$, such that $X_{nn'} = X_{n'n}$.
The power flow on line $(n,n')$ is denoted by $p_{nn'}$, such that \mbox{$p_{nn'} = -p_{n'n}$}.
The power injection at bus $n$ is denoted by $p_n$, and
the voltage phase angle at bus $n$ is denoted by $\theta_n$.
The last bus is selected as the \emph{reference bus}, so that \mbox{$\theta_{N+1} = 0$} and $\theta_i$, $i \in \{1,\ldots,N\}$ are expressed relative to the reference bus.

Although the relationship between the measurements and the phase angles at the buses is non-linear in an Alternating Current (AC) power system, it can be approximated using the linearized Direct Current (DC) power flow model. 
In particular, we obtain the following first-order approximation \cite{liu2014detecting,wood2013power} for the power flow on a single line: 
$
p_{nn'} = \frac{1}{X_{nn'}} (\theta_{n} - \theta_{n'}).
$
Using Kirchhoff's laws, we furthermore obtain that the power injected into bus $n$ equals the power outflow minus the inflow: 
$p_n = \sum_{n' \in \nb(n)} p_{nn'}$.
Let $x_{m}$ denote the measurement recorded by meter $m$. We assume that meter $m$'s measurement error is given by $\epsilon_{m}$, leading to the following equations:
$$
x_m = \left\{ \begin{array}{cl} 
p_{nn'} + \epsilon_{m} & \begin{array}{l}\text{if }m\text{ is a SCADA power} \\
\text{flow meter on line }(n,n'), \end{array} \\ 
p_{n} + \epsilon_m & \begin{array}{l}\text{if }m\text{ is a SCADA power} \\
\text{injection meter on bus }(n),\text{ and} \end{array} \\ 
\theta_{n} + \epsilon_m & \begin{array}{l}\text{if }m\text{ is a PMU for the voltage} \\
\text{phase angle on bus }(n). \end{array} 
\end{array}\right.
$$
The above equations can be combined into a single system of linear equations:
\begin{equation}
\vecx = \matb \vect + \vece,
\label{eq:main}
\end{equation}
which can be used for state estimation.
The topology matrix $\matb$ can then be constructed using the following algorithm:
\begin{enumerate}
\item Initialize $\matb$ as an $(M+N') \times N$ matrix filled with zeroes. Let $g_{mn}$ be its entry in the $m$th row and $n$th column.
\item For each bus injection meter $m$, let $n$ be the ID of its bus. For each $n' \in \nb(n)$, update \mbox{$g_{mn} \gets g_{mn} + \frac{1}{X_{nn'}}$} and \\ \mbox{$g_{mn'} \gets g_{mn'} - \frac{1}{X_{nn'}}$}.
\item For each power flow meter $m$, let $n$ be the ID of the first bus and $n'$ the ID of the second bus. Update \mbox{$g_{mn} \gets g_{mn} + \frac{1}{X_{nn'}}$} and \mbox{$g_{mn'} \gets g_{mn'} - \frac{1}{X_{nn'}}$}.
\item For each PMU $m$, let $n$ be the ID of its bus. Update \\ \mbox{$g_{mn} \gets g_{mn} + 1$}.
\item Zero-injection buses in $\nodes'$ are treated as bus injection meters that measure zero power, but which are given greater weight by multiplying the entries with a large constant $\largeconstant$. Let $n$ be the ID of its bus. For each $n' \in \nb(n)$, update \mbox{$g_{mn} \gets g_{mn} + \frac{\largeconstant}{X_{nn'}}$} and \mbox{$g_{mn'} \gets g_{mn'} - \frac{\largeconstant}{X_{nn'}}$}.
\end{enumerate}

\section{Threat Model Comparison}
\label{sec:threat_model_table}

In the earliest papers on FDI attack detection using state estimation \cite{bobba2010detecting,liu2011false}, the focus was on attackers who could either only compromise a fixed set of meters, or who had a maximum number $k$ of meters that they could compromise (but were free to choose which ones). This restriction was kept in some works \cite{yang2017optimal} and \cite{lakshminarayana2020data}, whereas in others, e.g., \cite{lakshminarayana2017optimal}, the attacker was instead restricted in terms of the total distortion that they could add to the measurements. Note that if an attacker has complete knowledge of the network topology \textit{and} the ability to change the measurements of \textit{all} meters to arbitrary values, then they are free to choose the measurements such that the residuals are zero.  Hence, any approach that uses state estimation to detect anomalies must assume that the attacker's capabilities are somehow limited. By contrast, we consider a system where the attacker is limited in terms of the number of OEMs that they can compromise, but an attack against a single OEM may affect all of a single operator's meters. An overview of these related threat models can be found in Table~\ref{tab:threat}.

\begin{table}[htp]
\caption{Overview of threat model restrictions in the literature.}
\begin{center}
\scalebox{0.8}{
\begin{tabular}{lc|cc}
& & \multicolumn{2}{c}{threat model} \\ 
& & limited & limited \\
& & $\#$ meters & total deviation \\ \toprule
Bobba et al. (2010) & \cite{bobba2010detecting} & $\checkmark$ & \\
Liu et al. (2011) & \cite{liu2011false} & $\checkmark$ & \\
Yang et al. (2017) & \cite{yang2017optimal} & $\checkmark$ & \\
Lakshminarayana et al. (2017) & \cite{lakshminarayana2017optimal} & & $\checkmark$\\
Lakshminarayana et al. (2020) & \cite{lakshminarayana2020data} & $\checkmark$ & \\
\end{tabular}
}
\end{center}
\label{tab:threat}
\end{table}%

\section{Numerical Example of the Evolution of Operators' Credits}
\label{sec:example}

To illustrate the evolution of the credits as discussed in Section~\ref{sec:incentives}, we consider a parameter setting based on Germany, where as of 2015 4 TSOs respectively controlled roughly $31.6\%$, $30.8\%$, $28.1\%$, and $9.5\%$ of the grid \cite{deloitte}. We choose the following parameters for our benchmark setting: $M = 1000$, $\initdeposit = 10^{14}$, $\reward = 10^6$, and $\misspen = \anompen = 10^{10}$. Here, $\initdeposit$ is chosen to be a large integer to avoid rounding errors, because the deposits are implemented as integers in Section~\ref{sec:implementation}. Furthermore, $\reward$ and $\misspen$ are chosen to ensure that it is profitable to add any meter whose probability of being offline in any slot is below $10^{-4}$. We also assume that all meters have the same error probabilities, namely $\missprob_m = 10^{-5}$ and $\anomprob_m = 0$ for all $m$.

\begin{figure}[h!]
     \centering
     \subfloat[][]{\includegraphics[width=.24\textwidth]{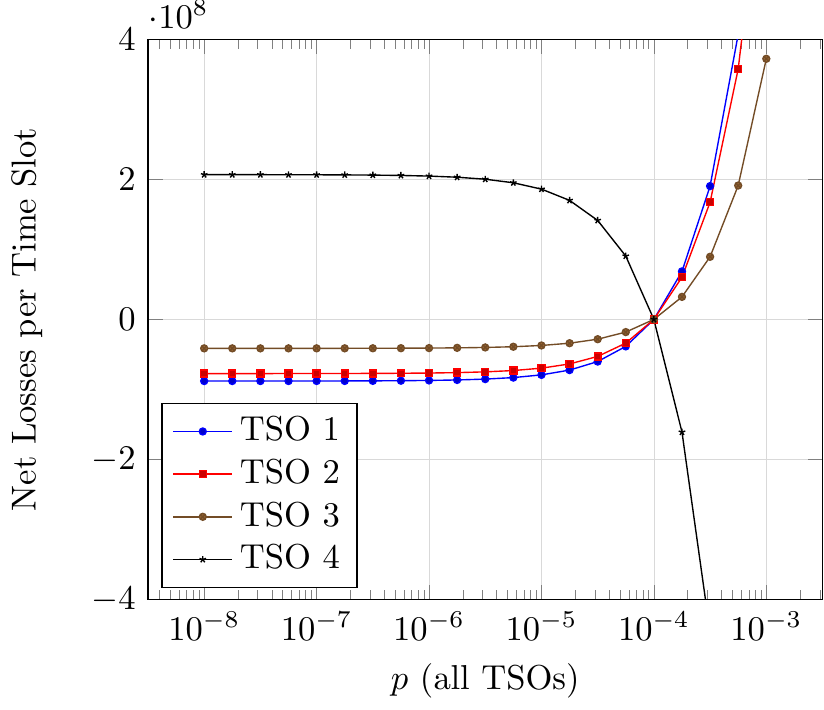}  
     \label{fig:plot1}}
     \subfloat[][]{\includegraphics[width=.24\textwidth]{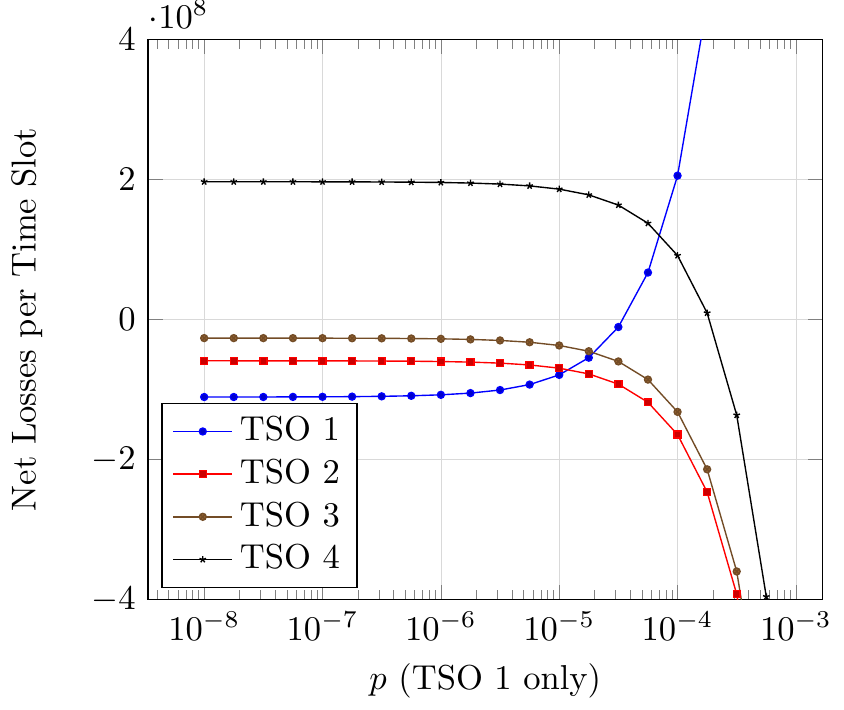}  
     \label{fig:plot2}}
     \caption{Net losses as a function of the probability of missing a measurement $p$. In Figure a), we choose $\misspen_m = p$ the same for all meters $m$ of the four TSOs. In Figure b), we only choose $\misspen_m = p$ if $m$ belongs to TSO 1, and $\misspen_m = 10^{-5}$ otherwise.}
     \label{fig:incentive_plots}
\end{figure}

In Figure~\ref{fig:plot1}, we display the net \textit{losses} per time slot for the four TSOs. At $\missprob_m = 10^{-5}$, we see that the smallest TSO loses roughly $2\cdot 10^{8}$ credits per time slot. This is divided almost equally among the other TSOs. With $\initdeposit = 10^{14}$, this means that the smallest TSO would run out of credits after $10^6$ time slots - if one time slot lasts 10 seconds, then this would be once every 115 days. The probability threshold $\zeta$ is close to $10^{-4}$. Once this threshold is crossed, meters become a liability and the smallest TSO is the only one that makes a profit.

In Figure~\ref{fig:plot2}, we keep $p_m$ constant at $10^{-5}$ for the meters of the last three TSOs, but vary the $\missprob$ for the meters of TSO 1. We see that the profits of TSO 1 approach a maximum of around $1.108\cdot10^{8}$ credits per time slot as $p_m$ decreases. However, if $p_m$ increases then the costs of TSO 1 eventually start to outweigh the gains. Eventually, the losses of TSO 1 become so high that even TSO 4 starts to make an expected profit each time slot -- however, at this point $p_m$ is higher than $\zeta$, which means that TSO 1 is better off joining the framework with a single meter.

\section{Incentive Compatibility Theorems}
\label{sec:theorems}

In this section, we prove the three theorems mentioned in Section~\ref{sec:incentive_properties}. We study the rewards of an operator $o \in \owners$ in a single time slot as a function of $o$'s actions. We first introduce five terms that determine the expected total change in $o$'s credits per time slot. These terms depend on $R_m = (1 - \missprob_{m}) \reward$ and $\Phi_m = \missprob_{m} \misspen$, which are the expected reward and penalty in each time slot for measurements produced and missed by meter $m$, respectively. 

The first term is the total expected reward for measurements produced by $o$'s meters, given by
$\rtot = \sum_{m \in \meters_o} R_m$.
The second term is the expected penalty for the missed measurements, given by 
$\phitot = \sum_{m \in \meters_o} \Phi_m$.
The third term represents the total losses from the rewards for other operators' meters, given by $\roth = \frac{1}{O-1} \sum_{m \in \meters \backslash \meters_o} R_m$. The fourth term represents the total gains from the missing measurement penalties for other operators, given by $\phioth = \frac{1}{O-1} \sum_{m \in \meters\backslash\meters_o} \Phi_m$. The fifth term represents the anomaly penalties and equals $\phip = \othersprob \sum_{m \in \meters_o}  \Phi'_m $, where $\othersprob$ equals the probability that $r > \anomthres$ and that every meter produced a measurement, i.e., $\prod_{m \in \meters} (1-\missprob_{m})$. The total expected change to $o$'s credits after each time slot is then given by:
\begin{equation}
\rnet = \rtot - \phitot - \roth + \phioth - \phip.
\label{eq:five_terms_raw}
\end{equation}
Incentive compatibility then means that given a set of possible actions, $o$ maximizes $\rnet$ by choosing the action recommended by the protocol. If this choice does not depend on the actions of the other operators, then it is a \textit{Nash equilibrium} for all operators to follow the protocol.

For the first theorem, we consider the impact of registering meters or not on $o$'s profits. Let $\iadd_m$ be equal to $1$ if meter $m$ has been added by its operator and $0$ otherwise. Let $\iiadd = (\iadd_m)_{m \in \meters}$ be the complete \textit{strategy profile} that describes for each meter whether it is added or not. The probability $\othersprob$ depends on which meters have been added or not, so it is a function of $\iiadd$. Similarly, the net gains $\rnet$ and its five constituent terms are functions of $\iiadd$.

For any $\iiadd \in \{0,1\}^{M}$, let $\iiadd_{m,0}$ and $\iiadd_{m,1}$ be equal to $\iiadd$ except that $\i_{m}$ has been set equal to $0$ or $1$, respectively. Incentive compatibility is then equivalent to the assertion that, for each $\iiadd$ and suitable $m$, it holds that $R_o^{\textsc{net}}(\iiadd_{m,1}) > R_o^{\textsc{net}}(I_{m,0})$ -- i.e., it is profitable for \textit{each} operator to add all suitable meters regardless of which meters are added by the other operators. As a consequence, all operators adding all suitable meters would be a Nash equilibrium because each operator would gain less by deviating. In Theorem~\ref{thm:addmeters}, we prove that as long as $\Delta_m(\iiadd) = \phip(\iiadd_{m,1}) - \phip(\iiadd_{m,0})$ is small, then it is profitable for all operators to add each meter $m$ for which $\misspen < \zeta$ for a given bound $\zeta$.
\begin{theorem}
For any meter $m$ such that $\missprob_m < \zeta_m(\iiadd)$, with 
$$
\zeta_m(\iiadd) \triangleq \frac{\reward}{\reward + \misspen} - \Delta_m(\iiadd),
$$
and for any $\iiadd \in \{0,1\}^{M}$ it holds that $R_o^{\textsc{net}}(\iiadd_{m,1}) > R_o^{\textsc{net}}(\iiadd_{m,0})$, where $\Delta_m(\iiadd) = \phip(\iiadd_{m,1}) - \phip(\iiadd_{m,0})$.
\label{thm:addmeters}
\end{theorem}

\begin{proof}
From \eqref{eq:five_terms_raw}, we observe that $\rnet(\iiadd_{m,1}) - \rnet(\iiadd_{m,0})$ can be decomposed into five terms as follows:
\begin{equation}
\begin{split}
\rnet(\iiadd_{m,1}) - \rnet(\iiadd_{m,0}) & = \rtot(\iiadd_{m,1}) - \rtot(\iiadd_{m,0}) \\
& - (\phitot(\iiadd_{m,1}) - \phitot(\iiadd_{m,0})) \\
& - (\roth(\iiadd_{m,1}) - \roth(\iiadd_{m,0})) \\
& + \phioth(\iiadd_{m,1}) - \phioth(\iiadd_{m,0}) \\
& - (\phip(\iiadd_{m,1}) - \phip(\iiadd_{m,0}))
\end{split}
\label{eq:five_terms}
\end{equation}
From the definition of these terms, we find that the first term equals $R_m$, the second term equals $\Phi_m$, and the third and fourth terms equal $0$. The fifth term represents the net effect of: 1) a higher likelihood that at least one measurement is missing, and 2) the impact on the residuals and therefore $r$ and $\anompen_m$. The latter effect depends on the place of the new meter in the grid and the degree to which the meter agrees with related meters. Whether this term, which equals $\Delta_m(\iiadd)$, is positive or negative therefore depends on the setting.

If we combine the above, we find that $\rnet$ is positive if and only if \mbox{$R_m - \Phi_m - \Delta_m(\iiadd) > 0$}, i.e.,
$$
\reward - \missprob_m \reward - \missprob_m \misspen - \Delta_{m}(\iiadd) > 0.
$$
which can be rewritten as
$$
 \missprob_m < \frac{\reward}{\reward + \misspen} - \Delta_{m}(\iiadd).
$$
 This proves the theorem.
\end{proof}

For the next theorem, we assume that it is known which meters have been registered or not, which are offline or not, and what the residuals would be based on the readings of all meters. We investigate the impact on the profits of an operator $o$ when the meters in the sets $\meters' = \meters'_o \cup \meters'_{\neg o}$ are withheld, such that \mbox{$\meters'_o \subset \meters_o$} and \mbox{$\meters'_{\neg o} \subset \meters \backslash \meters_o$}. We assume that the measurements in $\meters'$ would otherwise have made it onto the blockchain in time. Let $M'_o$ and $M'_{\neg o}$ be the number of elements in $\meters'_o$ and $\meters'_{\neg o}$, respectively. We assume that $o$ is aware which other meters have shared their measurements and what the values of these measurements are. This represents the situation where a single operator $o$ has the choice to withhold or block the measurements of the meters in $\meters'$ right before the end of the timeout period. We make these assumptions because we want to make the situation as advantageous as possible for the adversary for the corollaries that we derive from Theorem~\ref{thm:addmeasurements}.

In the following, let $\imiss_m$ be equal to $1$ if the measurement of meter $m$ is missed and $0$ if it is not. In this setting, the strategy profile $\iimiss$ is given by $(\imiss_m)_{m \in \meters}$. 
Let $\iimiss_{\meters',0}$ and $\iimiss_{\meters',1}$ be equal to $\iimiss$ expect that for all $m \in \meters'$, $\imiss_m$ has been set equal to $0$ or $1$, respectively. Since we assume that $o$ knows whether any meter has missed its measurement, it holds that $q = 0$ if at least one measurement is missing, and $q=1$ otherwise.  
Let \mbox{$\Delta_{\meters'}'(\iimiss) = \rnet(\iimiss_{\meters',0}) - \rnet(\iimiss_{\meters',1})$} be $o$'s profit from withholding the measurements of the meters in $\meters'$.
We can then prove the following statement.

\begin{theorem}

Let $\delta = \left(\frac{M'_{\neg o}}{O-1} - M'_o\right) \cdot (\reward + \misspen)$. Let ${\bf 1}_{M'}$  be equal to $1$ if $M'_o + M'_{\neg o}$ and to $0$ otherwise. Then
$$
\Delta'(\iimiss_{\meters'}) = \delta + q {\bf 1}_{M'} \sum_{m \in \meters_o} \Phi'_m.
$$

\label{thm:addmeasurements}
\end{theorem}

\begin{proof}
We can decompose $\rnet$ into five terms as in \eqref{eq:five_terms}, except that each quantity depends on $\iimiss_{\meters',0}$ and $\iimiss_{\meters',1}$ instead of $\iimiss_{m,0}$ and $\iimiss_{m,1}$. From the definition of $\rtot$, we find that the first term equals $- M_o' \reward$ as $o$ misses out on the rewards for the $M_o'$ meters in $\meters_o'$. The second term equals $- M_{o}' \misspen$, as the missed measurement penalty $\misspen$ is incurred an additional $M'_o$ times. The third equals $\frac{M'_{\neg o}}{O-1} \reward$ as $o$ loses fewer credits given as a reward to the other operators. Similarly, the fourth term equals $\frac{M'_{\neg o}}{O-1} \misspen$ as $o$ receives a fraction $\frac{1}{O-1}$ of the redistributed credits from the penalties incurred by the operators of the meters in $\meters'_{\neg o}$. The fifth term either equals $\sum_{m \in \meters_o}  \Phi'_m$ if $M'_o + M'_{\neg o} > 0$, as the anomaly penalty is no longer applied, and $0$ otherwise. By substituting these results into \eqref{eq:five_terms}, we find that 
$$
\Delta_{\meters'}'(\iimiss) = \left(\frac{M'_{\neg o}}{O-1} - M'_o\right) \cdot (\reward + \misspen) + q {\bf 1}_{M'} \sum_{m \in \meters_o}  \Phi'_m
$$
where ${\bf 1}_{M'}$ is as defined in the statement of the theorem. This completes the proof.

\end{proof}

In the final theorem, we derive an expression for the effect of distorting measurements on $o$'s profits. This does not have an impact on first four terms of \eqref{eq:five_terms_raw}, so we focus on $\phip$. If at least one measurement is missing, then the anomaly penalty is not applied at all, so we focus on the case where all meters have reported their measurements. As discussed in Section~\ref{sec:analysis}, $\vec{x}$ denotes the vector of ``true'' measurements, $\distortmeters$ the set of meters whose measurements can be arbitrarily distorted by $o$, and $\vec{d}$ is a $M \times 1$ vector such that $d_m$ is the distortion applied to meter $m$. In this setting, the residuals are a function of $\vec{d}$, i.e.,
$\hat{\epsilon}(\vec{d}) = \maker(\vec{x} + \vec{d})$,
where $\maker = I_{M+N'} - P$ where $I_{M+N'}$ is the identity matrix of size $I_{M+N'}$. In regression analysis, the matrix $\maker$ is often called the ``residual maker'' matrix, and we denote its $i,j$th entry by $\makerentry_{ij}$. The sum $r$ and the total anomaly penalty $\phip$ are also functions of $\vec{d}$, i.e.,
$$
\phip(\vec{d}) = \sum_{m \in \meters_o} \frac{\anompen}{r(\vec{d})}\left(\epsilon^2_{m}(\vec{d}) - \frac{r(\vec{d})}{M}\right).
$$
The total impact of adding distortions $\vec{d}$ on the anomaly penalty is then given by $\Phi'_{o}(\vec{d}) - \Phi'_{o}(\vec{0})$, where $\vec{0}$ is an $M \times 1$ vector of zeroes. The impact of adding distortions $\vec{d}$ on $o$'s net profits per time slot are positive if and only if
$\Phi'_{o}(\vec{d}) - \Phi'_{o}(\vec{0}) < 0$, i.e., if the anomaly penalties would be reduced. In Theorem~\ref{thm:distortmeasurements} we state a necessary and sufficient condition for this to hold. The theorem depends on the matrix $A$, given by 
$$
A = \begin{pmatrix} A_{dd} & A_{dx} \\ A_{xd} & {\bf 0}\end{pmatrix},
$$
where $A_{dd}$, $A_{dx}$, $A_{xd}$, and $\bf 0$ are $M \times M$ matrices. Each entry of $\bf 0$ equals zero, and the $k,l$th entry of $A_{dd}$, $A_{dx}$, and $A_{xd}$, denoted respectively by $a_{dd,kl}$, $a_{dx,kl}$, and $a_{xd,kl}$, are given by
$$
a_{dd,kl} = a_{dx,kl} = a_{xd,kl} = \sum_{i \in \meters_o} \left( \makerentry_{ik} \makerentry_{il} - \sum_{j \in \meters} \frac{\makerentry_{jk} \makerentry_{jl}}{M}\right).
$$
Finally, let $\vec{y} = \begin{pmatrix} \vec{d} \\ \vec{x}\end{pmatrix}$ be a $2M \times 1$ vector with the entries of $\vec{d}$ and $\vec{x}$.
\begin{theorem}
With $\vec{y}$ and $A$ as above, $\Phi'_{o}(\vec{d}) - \Phi'_{o}(\vec{0}) < 0$ if and only if
$$
\vec{y}^{\,T} A \vec{y} < 0.
$$
\label{thm:distortmeasurements}
\end{theorem}

\begin{proof}
First note that $\Phi'_{o}(\vec{d}) - \Phi'_{o}(\vec{0}) < 0$ is equivalent to saying that
\begin{equation}
 \sum_{m \in \meters_o}\left(\epsilon^2_{m}(\vec{d}) - \epsilon^2_{m}(\vec{0}) - \frac{1}{M} \left(r(\vec{d}) - r(\vec{0})\right)\right) < 0.
 \label{eq:rewards_full}
\end{equation}
Note that 
\begin{equation*}
\begin{split}
\epsilon^2_{m}(\vec{d}) & = \left(\sum_{k \in \meters_o} \makerentry_{mk} (x_k+d_k) \right)^2 \\ & = \sum_{k \in \meters_o}\sum_{l \in \meters_o}\makerentry_{mk} \makerentry_{ml} (x_k+d_k)(x_l+d_l)
\end{split}
\end{equation*}
so that
$$
\epsilon^2_{m}(\vec{d}) - \epsilon^2_{m}(\vec{0}) = \sum_{k \in \meters_o}\sum_{l \in \meters_o}\makerentry_{mk} \makerentry_{ml} (d_k d_l + d_k x_l + x_k d_l).
$$
Since $r(\vec{d}) = \sum_{k \in \meters} \epsilon^2_{k}(\vec{d})$, we can rewrite the left hand side of \eqref{eq:rewards_full} as a summation of terms $d_k d_l$, $d_k x_l$, and $x_k d_l$ for all $k,l \in \meters$. In particular, we can reorganize to produce the following matrix $A$
\begin{center}
\setlength{\tabcolsep}{3pt}
\begin{tabular}{c|cccccc}
      & $d_1$ & $\ldots$ & $d_M$ & $x_1$ & $\ldots$ & $x_M$ \\ \hline
$d_1$ & $a_{dd,11}$ & $\ldots$ & $a_{dd,1M}$ & $a_{dx,11}$ & $\ldots$ & $a_{dx,1M}$ \\
$\vdots$ & $\vdots$ & $\ddots$ & $\vdots$ & $\vdots$ & $\ddots$ & $\vdots$ \\
$d_M$ & $a_{dd,M1}$ & $\ldots$ & $a_{dd,MM}$ & $a_{dx,M1}$ & $\ldots$ & $a_{dx,MM}$ \\
$x_1$ & $a_{dd,11}$ & $\ldots$ & $a_{dd,1M}$ & $0$ & $\ldots$ & $0$ \\
$\vdots$ & $\vdots$ & $\ddots$ & $\vdots$ & $\vdots$ & $\ddots$ & $\vdots$ \\
$x_M$ & $a_{dd,M1}$ & $\ldots$ & $a_{dd,MM}$ & $0$ & $\ldots$ & $0$ \\
\end{tabular}
\setlength{\tabcolsep}{6pt}
\end{center}
for which it holds that $\vec{y}^{\,T} A \vec{y}$ corresponds to the left hand side of \eqref{eq:rewards_full}. This proves the theorem.
\end{proof}

\section{Blockchain Security}
\label{sec:blockchain_security}

\newcommand{\timeslotremainder}{y}

As we discussed in Section~\ref{sec:incentives}, it is typically profitable for an operator to prevent another operator to upload its measurements to the blockchain. 
The process of uploading a measurement from a meter to the blockchain is depicted in Figure~\ref{fig:data_flow}. 
Each measurement is created and sent by the meter over the network to its operator's MDMS server, which broadcasts it to the peer-to-peer network. The measurement must then be added to a block before \textit{FinalizeTimeSlot} -- this occurs after $\timeslotremainder$ blocks, where $\timeslotremainder \in \{0, \ldots,\slotduration+\timeoutblocks-1\}$ represents the number of rounds left  in which the round leader is aware of the measurement until the end of the current period and the timeout duration. After the measurement has been included in a block, this block must receive the support of at least two thirds of the nodes to be added to the blockchain. In Section~\ref{sec:model} we assume that it is beyond the capabilities of operators to attack other operator's meters, MDMS servers, or the network (e.g., through a DoS attack).
\footnote{For the same reason, we do not consider attacks in which an operator changes the reading of another operator's meter to an anomalous reading.} 
However, operators are able to influence decisions at the consensus layer. In the following, we discuss two attacks in which an operator delays the inclusion of a competing operator's readings beyond $\timeslotremainder$ blocks.

The first attack involves an adversarial coalition that refuses to endorse unwanted blocks. Since each block requires the endorsement of at least $2f+1$ peers to be committed, adversarial nodes have the possibility to veto blocks if abstaining would mean that more than $f$ peers do not endorse the block. For example, if there are $4$ companies involved, of which $1$ is adversarial and $1$ is offline for more than $\timeslotremainder$ blocks, then the adversarial node can veto all blocks by the two honest nodes, and only the blocks by the adversarial node will be committed. The adversarial node has full control over which transactions it includes in its block, so it can leave out the measurements of its competitors. However, although operators can periodically go offline, if $O$ is large then it is unlikely that many operators would go offline simultaneously.

The second attack involves an operator who is the round leader during all $\timeslotremainder$ blocks. Whether this possible depends on the choice of consensus protocol. In PBFT \cite{castro1999practical}, round leaders remain in place the same unless a node requests a view change. In this case, it is the responsibility of the maligned operators to request a view change if they observe that their measurements are not being included. However, in Tendermint \cite{kwon2014tendermint} and HotStuff \cite{yin2019hotstuff} rotate the leader every round. As such, each operator has to wait at most $O-1$ blocks until it can propose a block, so this is safe as long as $O - 1 < \timeslotremainder$. 
As such, $\timeoutblocks$ and $\slotduration$ should be chosen such that $\timeslotremainder > O-1$ with high probability.

\begin{figure}[t]
\centering
\includegraphics[width=0.45\textwidth]{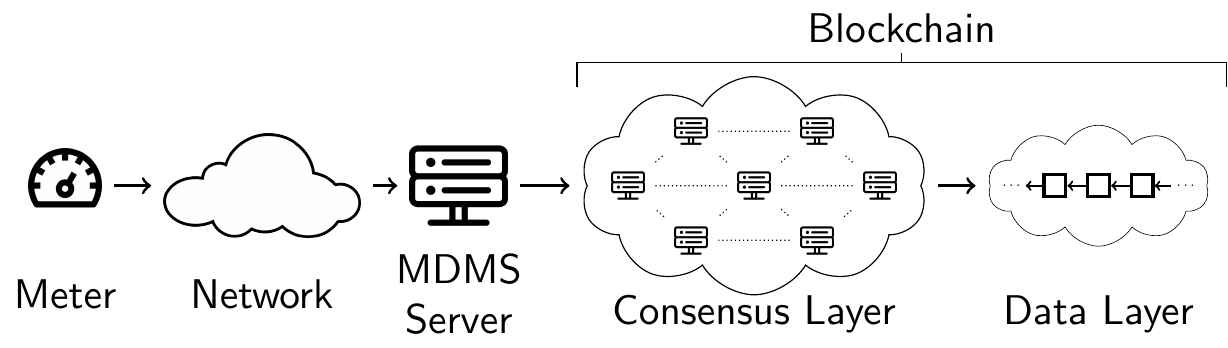}
\caption{Flowchart that represents the necessary stages for adding new measurements to the blockchain.}
\label{fig:data_flow}
\end{figure}

\end{document}